\documentclass[journal,10pt,twocolumn,doublespace]{IEEEtran}
\usepackage{amsmath,amssymb,epsfig,psfrag,cite,subfigure}
\include{macros}
\usepackage{amsmath}
\usepackage{cite}
\usepackage{url}
\usepackage{amsfonts}
\usepackage{amssymb}
\usepackage{graphicx,color}
\usepackage{graphicx}
\usepackage{epsfig,psfrag}
\usepackage{subfigure}%
\newtheorem{theorem}{Theorem}

\newtheorem{algorithm}{Algorithm}

\newtheorem{corollary}[theorem]{Corollary}

\newtheorem{lemma}[theorem]{Lemma}

\newtheorem{problem}[theorem]{Problem}
\newtheorem{proposition}[theorem]{Proposition}
\newtheorem{remark}[theorem]{Remark}

\newenvironment{proof}[1][Proof]{\noindent\textbf{#1.} }{\ \rule{0.5em}{0.5em}}

\begin{document}
\title{Cooperative Wireless Sensor Network Positioning via Implicit Convex Feasibility}

\author{{\IEEEauthorblockN{Mohammad Reza Gholami, Luba Tetruashvili, Erik G. Str\"om, and Yair Censor}}\\
  \thanks{
  Copyright (c) 2013 IEEE. Personal use of this material is permitted. However, permission to use this material for any other purposes must be obtained from the IEEE by sending a request to pubs-permissions@ieee.org. M. R. Gholami and E. G. Str\"om are with the Division of
    Communication Systems, Information Theory, and Antennas,
    Department of Signals and Systems, Chalmers University of
    Technology, SE-412 96 Gothenburg, Sweden. L. Tetruashvili and
    Y. Censor are with the Department of Mathematics, University of
    Haifa, Mt. Carmel, Haifa 3190501, Israel. The contributions of the
    first two authors to this work are of equal shares.}}
\maketitle

\begin{abstract}
  We propose a distributed positioning algorithm to estimate the
  unknown positions of a number of target nodes, given distance
  measurements between target nodes and between target nodes and a
  number of reference nodes at known positions.  Based on a geometric
  interpretation, we formulate the positioning problem as an implicit
  convex feasibility problem in which some of the sets depend on the
  unknown target positions, and apply a parallel projection onto
  convex sets approach to estimate the unknown target node positions.
  The proposed technique is suitable for parallel implementation in
  which every target node in parallel can update its position and
  share the estimate of its location with other targets.  We
  mathematically prove convergence of the proposed algorithm.
  Simulation results reveal enhanced performance for the proposed
  approach compared to available techniques based on projections,
  especially for sparse networks.

\textbf{Index Terms--}  Positioning, Cooperative wireless sensor network, Parallel projections onto convex sets, Convex feasibility problem, Implicit convex feasibility.
\end{abstract}

\section{Introduction}
For many applications in wireless sensor networks (WSNs), position
information is a vital requirement for the network to function as
intended.  Due to drawbacks of using global positioning system (GPS)
receivers in sensor nodes, mainly due to limited access to GPS
satellites, e.g., in an indoor scenario, the position recovery from
the network, called positioning or localization, has been extensively
studied in the
literature~\cite{Localization_algorithm_2009,Patrwari_2005_Locating_Nodes,Sayed_2005,Sinan_2005}.
It is usually assumed that there are a number of reference sensor
nodes at known positions that can be used to estimate the location of
a number of sensor nodes at unknown positions, henceforth called
target nodes. To estimate the position of target nodes, some types of
measurements are taken between different nodes such as
received-signal-strength, angle-of-arrival, or
time-of-arrival~\cite{Localization_algorithm_2009,Lic_Mohammad,Patrwari_2005_Locating_Nodes}.
A popular technique in the positioning literature is to estimate
distances between sensor nodes from measurements and then to apply a
suitable positioning algorithm~\cite{Localization_algorithm_2009}, and
this is also the approach taken in this study.

From one point of view, positioning algorithms can be categorized into
two groups: cooperative and noncooperative~\cite{Lic_Mohammad}.  In a
noncooperative network, measurements taken between a target and
reference nodes are used to estimate the position of the target node,
while in a cooperative network, besides the measurements made between
target nodes and reference nodes, measurements collected between
target nodes are also used in the positioning
process~\cite{Lic_Mohammad}. For low density networks, the cooperation
technique can effectively improve the performance of the position
estimate~\cite{Lic_Mohammad}.  Compared to noncooperative scenarios,
the positioning problem in a cooperative network is more challenging
and it is not clear how to effectively use the distance measurements
between target nodes for the positioning process.  During the last few
years, various cooperative positioning algorithms have been proposed 
in the literature.  Classic estimators such as
the maximum likelihood (ML) estimator and nonlinear least squares
(NLS) estimator derived for the positioning problem are often too
complex and pose difficult global optimization
problems~\cite{Lic_Mohammad,Patrwari_2005_Locating_Nodes,Gholami_SPAWC_2011}.
In the literature, a number of suboptimal techniques have been
proposed to avoid the difficulty in solving the ML or the NLS problems,
such as convex relaxation techniques, e.g., based on second order cone
programming~\cite{SOCP_Tseng_2007,SOCP_Srirangarajan_2008}, sum of
squares~\cite{Sum_of_squares}, and semidefinite
programming~\cite{Amir_LS_2008,Further_Relaxation_2008,SDP_Biswas_2006},
and linearization
techniques~\cite{Mitra_LS,Ho_2009_Succesive_Asymptotically,Gholami_ICC2011,Gholami_SPAWC_2011}.

From an implementation point of view, positioning algorithms can be
categorized into two classes: centralized and distributed. In a
centralized approach, all measurements gathered in the sensor nodes are
transferred to a central unit, and a positioning algorithm is applied
to find the locations of the target nodes. In a distributed approach,
however, every target is allowed to locally process the measurements
to obtain an estimate of its own position.  To get benefits from
cooperation, target nodes can share their estimates with other targets, e.g., by
broadcasting the estimates. For example, algorithms
based on weighted-multidimensional scaling \cite{Costa-2006} and
distributed belief propagation \cite{Belief_2005_localization} have
been proposed to solve the positioning problem in distributed
fashion.  In addition to the need for a distributed implementation,
another important factor in designing a positioning algorithm is the
robustness against non-line-of-sight (NLOS) conditions, in which 
distances are measured with a large, normally positive, errors.

It is well-known that in the noncooperative case, a particular target
node can be found in the intersection of a number of balls centered
around the reference nodes, if the measurement errors are positive. In
this case, it makes sense to formulate the positioning problem as a
convex feasibility problem (CFP), i.e., the target position is
estimated as a point inside the intersection of the appropriate
balls. We can approach this by finding a point that minimizes the sum
of squared distances to the balls. In case the intersection is
nonempty, a minimizing point is obviously a solution to the CFP. In
case the intersection is empty (i.e., when the CFP is inconsistent),
the minimizing points are still reasonable estimates of the target
node positions. There exist algorithms based on projection methods that
are proven to find a minimizer to the above-mentioned objective
function, see, e.g., \cite{Lic_Mohammad,Censor_97} and references
therein.


The class of \textit{projection methods} is understood here as the
class of methods that have the property that they can reach an aim
related to the family of sets $\{C_1,C_2,\ldots,C_m\}$, such as, but
not only, solving the CFP, or solving an optimization problem with
these sets as constraints, by performing projections (orthogonal,
i.e., least Euclidean distance, or others) onto the individual sets
$C_i$. The advantage of such methods occurs in situations where
projections onto the individual sets are computationally simple to
perform. Such methods have been in recent decades extensively
investigated mathematically and used experimentally with great success
on some huge and sparse real-wold applications, consult, e.g.,
\cite{Bauschke_1996,Censor_2012} and the books
\cite{Bauschke_2011,Byrne_2008,Cegielski_2012,Censor_97,Chinneck_2007,Escalante_2011,Galanti_2004,Herman_2009}.
For further applications of projection methods,
  see, e.g., \cite{Bauschke_2013}.

In this paper, we extend the projection ideas to the cooperative case. This leads to a new type of CFP, which we henceforth call an \emph{implicit} CFP (ICFP), since some of the convex sets (balls) actually depend on the target positions. We formulate an algorithm that minimizes the sum of the square distances to a number of convex sets. We also present a mathematical analysis to support the validity of our approach in terms of convergence. Simulation results show an improved performance for the proposed algorithm compared to available techniques based on projection approaches.

It should be noted that projection techniques have been used in the
past for cooperative positioning, see, e.g.,
\cite{Gholami_Eurasip_2011,Jia_Buehrer_2011}. In fact, our proposed
scheme can be viewed as an extension of the parallel projection method
(PPM) proposed in \cite{Jia_Buehrer_2011}. However, our method is
different from the one in \cite{Jia_Buehrer_2011} in terms of
convergence.  The approach proposed in \cite{Jia_Buehrer_2011} needs
good initial points to provide good estimates of the targets positions,
while for the proposed technique in this paper, convergence to an optimal solution is
always guaranteed for any arbitrary initialization point.
There are no convergence proofs for
  the algorithms proposed in
  \cite{Gholami_Eurasip_2011,Jia_Buehrer_2011} and one aim of this
  study is to provide a framework for studying the convergence of a
  class of optimization problems with applications in wireless sensor
  network positioning.

%
%
%
%

The paper is organized as follows. In Section\,\ref{sec:sys_model}, we
describe the positioning problem based on a geometric notion. The
proposed algorithm is defined in Section~\ref{sec:algo}, and a
mathematical model of the ICFP problem is provided in
Section~\ref{eq:math_model}. Some properties of the objective function
and the convergence proof are found in
Sections~\ref{sec:prop_costfunc} and~\ref{sec:conv_anal},
respectively.  In Section~\ref{sec:sim_resul}, the performance of the
proposed algorithm is evaluated by computer simulations.

\section{Problem statement}
\label{sec:sys_model}
\subsection{System model}

We consider a $d$-dimensional cooperative network, $d=2$ or $d=3$,
consisting of a number of sensor nodes at known positions, called as
reference nodes, and a number of target nodes at unknown positions. In
this study, we assume that every target node can estimate the distance
to nearby sensor nodes, e.g., using a two-way time-of-arrival
technique.  We also assume that sensor nodes are perfectly
synchronized, meaning that measurements are no longer affected by
clock imperfections.  It is also assumed that target nodes can share
the estimates of the positions by broadcasting the estimates to
neighboring target nodes (without any round-off errors).

Let us consider a WSN consisting of $n+m$ sensor nodes distributed in
a space.  Let $x_{i}\in\mathbb{R}^{d}$ for $i\in I:=\{1,2,\ldots,n\}$
be the position of target nodes whose locations $x_{i}$ are unknown,
and let $a_{j}\in\mathbb{R}^{d}$ for $j\in J:=\{n+1,n+2,\ldots,n+m\}$
be the position of reference nodes whose locations $a_{j}$ are
\emph{a~priori} known.
To formulate a geometric positioning problem, we further define the
following sets.  For every $i\in I,$ let%
\begin{align}
\mathcal{A}_{i}=\{j\in J\mid &\text{ target node $i$ can communicate with}\nonumber\\
&~~\text{reference node } {j} \}
\end{align}
and let%
\begin{align}
\mathcal{B}_{i}=\{q\in I\mid &\text{ target node $i$ can communicate with target}\nonumber\\
&~~\text{node } {q}\}.
\end{align}
The distance measurement between a pair of nodes (target, reference)
or (target, target) is modeled as in~\cite{Sayed_2005,Sinan_Survey},
\begin{IEEEeqnarray}{rCl}\label{eq:dis_estimate}
\IEEEyesnumber\IEEEyessubnumber
&\hat{d}_{ij}=d_{ij}+\varepsilon_{ij},~j\in\mathcal{A}_i,~i\in I\IEEEyessubnumber\\
&\hat{l}_{iq}=l_{iq}+\varepsilon_{iq},~q\in\mathcal{B}_i,~i\in I\IEEEyessubnumber
\end{IEEEeqnarray}
where $\varepsilon_{ij}$ and $\varepsilon_{iq}$ are measurement errors
drawn from some distributions, and $d_{ij}=\|x_i-a_j\|$ and
$l_{iq}=\|x_i-x_q\|$ are the actual Euclidean distances between
reference node $j$ and target node $i$ and between target node $q$ and
target node $i$, respectively.  We also define the following convex
sets.  For $i\in I,$ define for each $j\in \mathcal{A}_{i},$
\begin{equation}
\label{eq:set1}
\mathcal{C}_{ij}=\{z\in\mathbb{R}^{d}\mid\left\Vert z-a_{j}\right\Vert
\leq\hat{d}_{ij}\}
\end{equation}
and for $i\in I,$ define for each $q\in \mathcal{B}_{i},$
\begin{equation}
\label{eq:set2}
\mathcal{X}_{iq}=\{z\in\mathbb{R}^{d}\mid\left\Vert z-x_{q}\right\Vert
\leq\hat{l}_{iq}\}.
\end{equation}

For convenience, we assume that the measurements are symmetric in the
sense that if $q\in\mathcal{B}_i$ then $i\in\mathcal{B}_q$ and
furthermore that $\hat l_{iq}=\hat l_{qi}$. That is, if target $i$ can measure
the distance to target $q$ then the opposite is also true. This
simplifying assumption can be motivated by extending the traditional two-way
time-of-arrival (TOA) ranging procedure to a four-way TOA process in
which
\begin{enumerate}
\item Node $i$ sends a message to node $q$
\item Node $q$ immediately returns the message to node $i$
\item Node $i$ can now  compute the initial estimate $\hat l_{iq}'$ from the TOA measurements recorded from the two first transmissions
\item Node $i$ transmits  $\hat l_{iq}'$ to node $q$
\item Node $q$ can now compute the initial estimate $\hat l_{qi}'$
  from the TOA measurements from the second and third transmissions,
  and the final estimates are computed as ${\hat l_{qi} = \hat l_{iq}
    = (\hat l_{iq}' + \hat l_{qi}')/2}$
\item Node $q$ transmits the final estimate $\hat l_{iq}$ to node $i$.
\end{enumerate}

From Eq. \eqref{eq:dis_estimate} it is clear that in the absence of
measurement errors, i.e., when $\varepsilon_{ij}=0$ and
$\varepsilon_{iq}=0$, the position of target node $i$ belongs to the
intersection of a number of closed balls centered at $a_j$ and $x_q$
with radii $d_{ij}$ and $l_{iq}$, respectively. Assuming positive
measurement errors, we use the balls defined in Eqs.~\eqref{eq:set1}
and \eqref{eq:set2}, and formulate a
feasibility problem (CFP) aimed at solving the cooperative
sensor network positioning problem as follows. 

\begin{problem}
\label{prob:sim}Find a set of vectors $x_{i}\in\mathbb{R}^{d},$ $i\in
I,$ such that
\begin{equation}
x_{i}\in\left(  \bigcap\limits_{j\in \mathcal{A}_{i}}\mathcal{C}_{i j}\right)
\bigcap\left(  \bigcap_{q\in \mathcal{B}_{i}}\mathcal{X}_{i q}\right)
.\label{eq:problem}%
\end{equation}
\end{problem}

Note that when the problem is formulated, the target nodes
$x_{i}\in\mathbb{R}^{d}$ for $i\in I,$ are not known, and thus the
sets $\mathcal{X}_{iq}$ are also unknown. For more details and other
examples of geometric positioning problems, see
\cite{Lic_Mohammad,Gholami_Eurasip_2011}.

\begin{figure}
\centering
 \psfrag{r3}[cc][][1]{$a_3$}
 \psfrag{r4}[cc][][1]{$a_4$~~}
 \psfrag{r5}[cc][][1]{$a_5$}
 \psfrag{r6}[cc][][1]{$a_6$}
 \psfrag{t1}[cc][][1]{$x_1$}
 \psfrag{t2}[cc][][1]{$x_2$}
 \psfrag{target}[c][][.8]{~~~target node }
 \psfrag{reference}[c][][.8]{reference node}
 \psfrag{sets}[c][][.8]{$\mathcal{A}_1=\{3,4\},~\mathcal{A}_2=\{5,6\},~\mathcal{B}_1=\{2\},~\mathcal{B}_2=\{1\}$}
  \includegraphics[width=90mm]{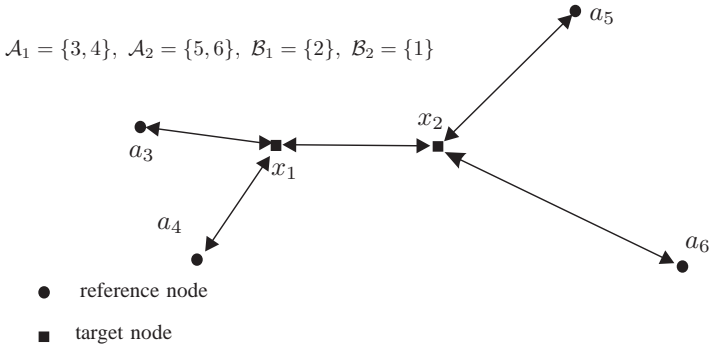}
  \caption{An example of a cooperative network consisting of two target nodes and four reference nodes.}
 \label{fig:twotargets}
\end{figure}
\begin{figure}
\centering
 \psfrag{r3}[cc][][1]{$a_3$}
 \psfrag{r4}[cc][][1]{$a_4$~~}
 \psfrag{r5}[cc][][1]{$a_5$}
 \psfrag{r6}[cc][][1]{$a_6$}
 \psfrag{t1}[cc][][1]{$x_1$}
 \psfrag{t2}[cc][][1]{$x_2$}
 \psfrag{d13}[cc][][1][90]{$\hat{d}_{13}$}
 \psfrag{d14}[cc][][1]{$\hat{d}_{14}$~}
 \psfrag{d25}[cc][][1]{$\hat{d}_{25}$~~}
 \psfrag{d26}[cc][][1]{$\hat{d}_{26}$~~}
 \psfrag{d12}[cc][][1]{$\hat{l}_{12}$}
 \psfrag{d21}[cc][][1]{$\hat{l}_{21}$}
  \psfrag{target}[c][][.8]{\qquad target node }
 \psfrag{reference}[c][][.8]{\quad~~reference node}
  \includegraphics[width=90mm]{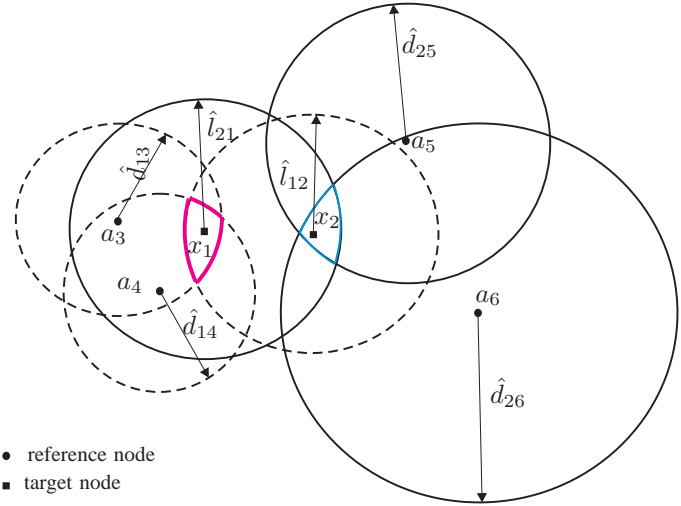}
  \caption{Deriving the intersections for target node one and target node two.
  Balls with dashed boundaries are used to derive the intersection for target node one and
  balls with solid boundaries are used to determine the intersection for target node two.
  The boundary of intersections for target node one and target node two are red and blue colored, respectively.
  Here, we assume that $\mathcal{X}_{12}$ and $\mathcal{X}_{21}$ are available for finding the red and blue areas.
  In a practical setting, however, $\mathcal{X}_{12}$ and $\mathcal{X}_{21}$ are not available from the start.}
 \label{fig:Intersection_twotargets}
\end{figure}

We clarify the formulation of Problem 1 with the following
example. Consider Fig.\,\ref{fig:twotargets} in which a 2-dimensional
cooperative network consists of four reference nodes, i.e.,
${a_i\in\mathbb{R}^2,~i\in\{3,4,5,6\}}$, and two target nodes,
$x_{i}\in\mathbb{R}^2,~i\in\{1,2\}$.  The sets $\mathcal{A}_1$,
$\mathcal{A}_2$, $\mathcal{B}_1$, and $\mathcal{B}_2$ for this network
are shown in the figure. The sets $\mathcal{C}_{ij}$ and
$\mathcal{X}_{i q}$ are defined as
\begin{IEEEeqnarray}{rCl}
\IEEEyesnumber\IEEEyessubnumber
&\mathcal{C}_{13}=\{z\in \mathbb{R}^2~|~\|z-a_3\|\leq \hat{d}_{13}\},\nonumber\\
&\mathcal{C}_{14}=\{z\in \mathbb{R}^2~|~\|z-a_4\|\leq \hat{d}_{14}\},\nonumber\\
&\mathcal{C}_{25}=\{z\in \mathbb{R}^2~|~\|z-a_5\|\leq \hat{d}_{25}\},\nonumber\\
&\mathcal{C}_{26}=\{z\in \mathbb{R}^2~|~\|z-a_6\|\leq \hat{d}_{26}\},\nonumber\\
&\mathcal{X}_{12}=\{z\in \mathbb{R}^2~|~\|z-x_2\|\leq \hat{l}_{12}\},\nonumber\\
&\mathcal{X}_{21}=\{z\in \mathbb{R}^2~|~\|z-x_1\|\leq \hat{l}_{21}\}.
\end{IEEEeqnarray}
{Suppose that distance estimates between different
  nodes are greater than the actual distances, i.e.,
  ${\hat{d}_{ij}\geq d_{ij}}$ for $i\in \{1,2\}$ and $j\in
  \{3,4,5,6\}$, and $\hat{l}_{i q}\geq l_{i q}$ for $i\in \{1,2\}$ and
  $q\in \{1,2\}$. Such a situation happens when the measurement errors
  are positive \cite{Gholami_upperbound_2012,Henk_2012_TC}.}  The intersection
involving target node one and the intersection involving target node
two are shown in Fig.\,\ref{fig:Intersection_twotargets}.  Hence, we
can write
\begin{align}
&x_1\in \Big(\mathcal{C}_{13}\bigcap \mathcal{C}_{14}\Big)\bigcap\mathcal{X}_{12},\\
&x_2\in \Big(\mathcal{C}_{25}\bigcap \mathcal{C}_{26}\Big)\bigcap\mathcal{X}_{21}.
\end{align}

{\begin{remark}
\label{rem:positive_error}
In practical applications, some of the measurement errors $\epsilon_{ij}$ and
$\epsilon_{iq}$ in \eqref{eq:dis_estimate} can be negative. 
In this case, the intersection \eqref{eq:problem} might not contain 
the target node locations, even if the intersection is nonempty. However, there
are also practical situations when the measurement errors are
positive. For example, errors in two-way time-of-arrival (TW-TOA)
measurements tend to be positive in practice, even for line-of-sight
conditions (as clarified in recent work on localization based on
practical Ultra-wideband (UWB) measurements, please see\cite{Henk_2012_TC}).  
Time-of-arrival is typically measured by computing the cross-correlation of
the received signal with a copy of the transmitted signal and finding
the first peak of the cross-correlation that exceeds
a certain threshold. Even if the threshold is carefully adjusted, the detected
peak rarely occurs before the true arrival time, especially for UWB signals. 
For more details of this
phenomenon, see~\cite{Henk_2012_TC}. 
Finally, note that we have here considered positive errors to justify the ICFP
as a means to estimate the target node positions. However, the algorithm
proposed later in this paper does not rely on this assumption. It will
function well even if errors are allowed to be negative.
\end{remark}}

\subsection{On the \textquotedblleft implicit convex feasibility
problem\textquotedblright\label{sect:implicit}}
\label{sec:ICFP}
Problem 1 is not a standard CFP because, as can be seen from
(\ref{eq:problem}), the sets $\mathcal{X}_{i q}$ are not determined
only by the data, but also on the, yet unknown, target node positions.

For CFPs with fixed feasible sets that are not dependent on the
solution there exits a broad literature, see, e.g.,~\cite{Censor_97}
or~\cite{Bauschke_Borwein_1996}.

This leads us to formulate a new class of CFPs that we call implicit CFP (ICFP).
To start, 
let $Q_{t}$ for $t=1,2,\ldots,T$ be subsets of $\mathbb{R}^{p}$ defined
for given convex functions $f_{t}:\mathbb{R}^{p}\rightarrow
\mathbb{R}$ by their zero-level-sets:
\begin{equation}\label{eq1}
    Q_{t}:=\{z\in \mathbb{R}^{p} \mid f_{t}(z)\leq 0\}.
\end{equation}

\emph{The convex feasibility problem} (CFP) in the finite-dimensional Euclidean space $\mathbb{R}^{p}$ is defined in the next Problem~2:

\emph{Problem 2}: Find a point $z^{*}$ in the set $Q:=\cap_{t=1}^{T}Q_{t}$.

Now let us consider the implicit convex feasibility problem (ICFP),
that generalizes the CFP, and which can be used to describe the
cooperative positioning problem.

Assume that $h_{t}: \mathbb{R}^{p}\times\mathbb{R}^{p}\rightarrow
\mathbb{R}$ ($t=1,2,\ldots,T$) are some given functions, such that for
any point $u\in \mathbb{R}^{p}$, the function $h_{t}(z,u)$ is convex
in the variable $z$.  For any given point $u\in \mathbb{R}^{p}$ and
$t=1,2,\ldots,T$ we define the sets
\begin{equation}\label{eq2}
  \hat{Q}_{t}(u):=\{z\in \mathbb{R}^{p}|h_{t}(z,u)\leq 0\},
\end{equation}
and look at the following ICFP:

\emph{Problem 3:} Find a point $z^{*}$ in the set
$\hat{Q}(z^{*}):=\cap_{t=1}^{T}\hat{Q}_{t}(z^{*})$.

In the ICFP problem we are seeking a common point $z^{*}$ of sets
which are defined by the unknown point $z^{*}$. Choosing
$h_{ij}(z,a_{j})=\| z-a_{j}\|-\hat{d}_{ij}$, for all $i \in I$ and
$j\in \mathcal{A}_{i}$, and
$h_{iq}(z,x_{q})=\|z-x_{q}\|-\hat{l}_{iq}$, for all $i \in I$ and
$q\in \mathcal{B}_{i}$, the ICFP translates the cooperative
positioning problem into the following problem.


\emph{Problem 4:} Find a point $(x_{1},x_{2},\ldots,x_{n})$ in
$\mathbb{R}^{dn}$ such that it satisfies following relations for all
$i=1,2,\ldots,n$:
\begin{align}
\label{eq:CFP_positioning}
x_{i} \in &\left(\bigcap_{j\in \mathcal{A}_{i}}\{z\in \mathbb{R}^{d}
  \mid \| z-a_{j}\|-\hat{d}_{ij}\leq 0 \}\right)\bigcap\nonumber\\
&\left(\bigcap_{q\in \mathcal{B}_{i}}\{z\in \mathbb{R}^{d} \mid \| z-x_{q}\|-\hat{l}_{iq}\leq 0 \}\right).
\end{align}

\section{The algorithm}
\label{sec:algo}
To solve Problem~4, as described in the previous section, we propose a
distributed algorithm based on projections onto sets
$\mathcal{C}_{ij}$ and $\mathcal{X}_{iq}$, defined in \eqref{eq:set1}
and \eqref{eq:set2}.  The notation $P_{\Omega}(u)$ stands for the
orthogonal (least Euclidean distance) projection of the point $u$ onto
the set $\Omega$. We use boldface to denote vectors in the product
space $\mathbb{R}^{dn}$ ($d=2$ or $d=3$), and denote the number of
elements in a set $D$ by $|D|$.

\begin{algorithm}
\label{alg:alg}$\left.  {}\right.  $

\noindent{\textbf{Initialization:}} Choose an arbitrary initial point
${\boldsymbol{\mathbf{x}}^{0}=(x_{1}^{0},x_{2}^{0},\ldots,x_{n}^{0})\in
\mathbb{R}^{dn}.}$

\noindent\textbf{General Iterative Step:} Given a current vector
${\boldsymbol{\mathbf{x}}^{k}=(x_{1}^{k},x_{2}^{k},\ldots,x_{n}^{k})\in
\mathbb{R}^{dn},}$

\begin{enumerate}
\item for every $i=1,2,\ldots,n,$ define%
\begin{equation}
\mathcal{X}_{iq}^{k}:=\left\{
\begin{array}
[c]{ll}%
\{z\mid\Vert z-x_{q}^{k}\mid \leq \hat{l}_{iq}\}, & \text{for
}i<q,\text{ }\\
\{z\mid\Vert z-x_{q}^{k+1}\mid \leq \hat{l}_{iq}\}, & \text{for }i>q,
\end{array}
\right.
\end{equation}
and calculate%
\begin{align}
x_{i}^{k+1}=\frac{1}{2}x_{i}^{k}&+\frac{1}{2(|\mathcal{A}_{i}|+|\mathcal{B}%
_{i}|)}\Bigg(  \sum_{j\in\mathcal{A}_{i}}P_{\mathcal{C}_{ij}}(x_{i}^{k})\nonumber\\
&+\sum
_{q\in\mathcal{B}_{i}}P_{\mathcal{X}_{iq}^{k}}(x_{i}^{k})\Bigg)  , \label{eq:nupdates}%
\end{align}

\item set%
\begin{equation}
\boldsymbol{\mathbf{x}}^{k+1}=(x_{1}^{k+1},x_{2}^{k+1},\ldots,x_{n}^{k+1}).
\end{equation}

\end{enumerate}

\noindent\textbf{Stopping criterion:} If $\|\boldsymbol{\mathbf{x}%
}^{k+1}-\boldsymbol{\mathbf{x}}^{k}\|$ is small enough then stop.
\end{algorithm}

Note that in the general iterative step there are $n$ update-steps
(\ref{eq:nupdates}) inside the $k$-th iteration, leading from
$\boldsymbol{\mathbf{x}}^{k}$ to $\boldsymbol{\mathbf{x}}^{k+1}$.  In
the $i$th update step, we determine a new approximate location $x_{i}%
^{k+1\text{ }}$of the unknown target $x_{i}$ by computing projections
onto balls around reference points and around target points for $q\neq
i$ and then taking the mid-point on the line segment between the
average of projections and the old approximate location $x_{i}^{k}$ of
the target $x_{i}$.

%

\section{The mathematical model of the problem}
\label{eq:math_model}
Let us consider the mathematical model of the problem that is obtained
by unconstrained minimization%
\begin{equation}
\left\{
\begin{array}
[c]{ll}%
\mathrm{minimize} & f({\mathbf{x}})\,\\
\mathrm{subject~ to} & {\mathbf{x}}\in \mathbb{R}^{dn}%
\end{array}
\right.  \label{eq:minprob}%
\end{equation}
of the objective function $f:\mathbb{R}^{dn}\rightarrow \mathbb{R}$
given by:%
\begin{align}
f({\mathbf{x}})\,:=&\,\sum_{i=1}^{n}\sum_{j\in\mathcal{A}_{i}}%
(\max\{\| x_{i}-a_{j}\|-\hat{d}_{ij},0\})^{2}\nonumber\\
&+\frac{1}{2}\sum_{i=1}^{n}\sum
_{q\in\mathcal{B}_{i}}(\max\{\| x_{i}-x_{q}\|-\hat{l}_{iq},0\})^{2}
\label{eq:function1}%
\end{align}
for $\boldsymbol{\mathbf{x}}=(x_{1},x_{2},\ldots,x_{n})$ and $x_{i}\in
\mathbb{R}^{d}$ for all $i=1,2,\ldots,n.$ The factor $1/2$ in (17) is
motivated by the fact that each term in the sum is repeated twice, due
to the symmetry of the measurements ($q\in\mathcal{B}_i \Rightarrow
i\in\mathcal{B}_q$ and $\hat l_{iq}=\hat l_{qi}$).  The objective
function $f$ is the sum of convex functions and is therefore
convex. Obviously, the optimal value of this optimization problem is
zero if we assume that there exists a point which satisfies all
distance requirements. Under the last mentioned existence assumption,
the motivation for introducing the optimization problem
\eqref{eq:minprob} is that any optimal solution of the problem
(\ref{eq:minprob})-(\ref{eq:function1}) is a solution of the ICFP
Problem 4.


Note that the objective function can also be rewritten equivalently in the
following form:
\begin{align}
&f(\boldsymbol{\mathbf{x}})\nonumber\\
&=\sum_{i=1}^{n}\sum_{j\in\mathcal{A}_{i}}\|
x_{i}-P_{\mathcal{C}_{ij}}(x_{i})\|^{2}+\frac{1}{2}\sum_{i=1}^{n}\sum_{q\in\mathcal{B}_{i}%
}\| x_{i}-P_{\mathcal{X}_{iq}}(x_{i})\|^{2}, \label{problem}%
\end{align}
where $P$ is the orthogonal projection operator and the sets $\mathcal{C}_{ij}$ and
$\mathcal{X}_{iq}$ are defined in \eqref{eq:set1} and \eqref{eq:set2}, respectively.

We will later show that any sequence
$\{\boldsymbol{\mathbf{x}}^{k}\}_{k=0}^{\infty}$ generated by
Algorithm \ref{alg:alg} is such that its accumulation points solve
(\ref{eq:minprob}) with the objective function (\ref{problem}).
This convergence property holds regardless if the ICFP in
\eqref{eq:CFP_positioning} is consistent or not.

{\begin{remark} There are related approaches in the
    literature to solve the positioning problem.  For example, the
    authors in \cite{Enyang_2011} consider a mini-max approach and
    obtain a positioning problem based on convex relaxation. The
    nonconvex problem originating from the least squares criterion can
    be transformed to a convex problem by adopting, e.g., a convex
    relaxation~\cite{SDP_Biswas_2006}. For sparse networks, the
    author of\cite{Sum_of_squares} formulated a positioning algorithm 
    based on sum of squares and derived a convex optimization
    problem. The main difference between most of the available methods
    in the literature and the algorithm developed in this paper, is
    that the algorithm proposed here is a projection method that has several advantages as discussed, e.g., in~\cite[Subsection 3.2]{Bauschke_2013}.
\end{remark}}

\section{Some properties of the objective function}
\label{sec:prop_costfunc}
The following proposition contains some important properties of the objective
function $f$ that are required for the convergence analysis in the sequel.

\begin{proposition}
\label{prop1}The function $f$ of (\ref{problem}) is convex, continuously
differentiable, and its gradient is block-coordinate-wise Lipschitz continuous
with constants $L_{i}\triangleq 4(|\mathcal{A}_{i}|+|\mathcal{B}_{i}|)$ for all
$i=1,2,\ldots,n.$
\end{proposition}

\begin{proof}
\noindent Convexity of the function $f$ follows from its representation in
(\ref{eq:function1}) where each summand is a composition of a square function
and a nonnegative convex max-type function and thus $f$ is convex.
Let us consider the following equivalent representation of \eqref{problem}:
$$
f(\boldsymbol{\mathbf{x}})=\sum_{i=1}^{n}\sum_{j\in\mathcal{A}_{i}}\|
x_{i}-P_{\mathcal{C}_{ij}}(x_{i})\|^{2}+\sum_{q\in\mathcal{B}_{s}%
}\| x_{s}-P_{\mathcal{X}_{sq}}(x_{s})\|^{2}$$
\begin{equation}
-\frac{1}{2}\sum_{q\in\mathcal{B}_{s}
}\| x_{s}-P_{\mathcal{X}_{sq}}(x_{s})\|^{2}+\frac{1}{2}\sum_{i=1, i\neq s}^{n}\sum_{q\in\mathcal{B}_{i}%
}\| x_{i}-P_{\mathcal{X}_{iq}}(x_{i})\|^{2}. \label{func_i}%
\end{equation}
Note that for each pair of target neighbours $i,q$ the following identity is satisfied:
\begin{equation}\label{targets}
    \| x_{i}-P_{\mathcal{X}_{iq}}(x_{i})\|=\| x_{q}-P_{\mathcal{X}_{qi}}(x_{q})\|.
\end{equation}
Substituting (\ref{targets}) into (\ref{func_i}) we get:
\begin{align}
f(\boldsymbol{\mathbf{x}})=&\sum_{i=1}^{n}\sum_{j\in\mathcal{A}_{i}}\|
x_{i}-P_{\mathcal{C}_{ij}}(x_{i})\|^{2}+\sum_{q\in\mathcal{B}_{s}%
}\| x_{s}-P_{\mathcal{X}_{sq}}(x_{s})\|^{2}\nonumber\\
&+\frac{1}{2}\sum_{i=1, i\neq s}^{n}\sum_{q\in\mathcal{B}_{i},q\neq s%
}\| x_{i}-P_{\mathcal{X}_{iq}}(x_{i})\|^{2}. \label{func_1ii}%
\end{align}
Thus, the partial derivative of
$f$ at $(x_{1},x_{2},\ldots,x_{n})$ with respect to the variables in the
subvector $x_{s}$ is given by:
\begin{align}\label{func_ii}
\nabla_{s}f(x_{1},x_{2},\ldots,x_{n})=&\sum_{j\in\mathcal{A}_{s}}%
2(x_{s}-P_{\mathcal{C}_{sj}}(x_{s}))\nonumber\\
&+\sum_{q\in\mathcal{B}_{s}}2(x_{s}-P_{\mathcal{X}_{sq}}%
(x_{s})),
\end{align}
and the gradient of $f$ is then given by
\begin{align}
&\nabla f(x_{1},x_{2},\ldots,x_{n})\nonumber\\
&=\Big(\nabla_{1}f(x_{1},x_{2},\ldots
,x_{n}),\ldots,\nabla_{n}f(x_{1},x_{2},\ldots,x_{n})\Big).
\end{align}
From continuity of $\nabla_{s}f(x_{1},x_{2},\ldots,x_{n})$ for all
${s=1,2,\ldots,n}$, we conclude that the gradient of $f$ is continuously differentiable.

It remains to show that gradient is block-coordinate-wise Lipschitz
continuous. Let us take two vectors $(x_{1},x_{2},\ldots,y_{i},\ldots,x_{n})$
and $(x_{1},x_{2},\ldots,x_{i},\ldots,x_{n})$. These two vectors are identical
except for elements which correspond to the $i$-th subvector and
we have%
\begin{align*}
&\Vert\nabla_{i}f(x_{1},x_{2},\ldots,y_{i},\ldots,x_{n})-\nabla_{i}%
f(x_{1},x_{2},\ldots,x_{i},\ldots,x_{n})\Vert\nonumber\\
&\leq4(|\mathcal{A}_{i}%
|+|\mathcal{B}_{i}|)\Vert x_{i}-y_{i}\Vert,
\end{align*}
where in the last inequality, we used the nonexpansivity of the projection operator
$P$, i.e.,
\[{\|P(x)-P(y)\| \le \|x-y\|}.\]
Thus, we have shown that the gradient of $f$ is block-coordinate-wise
Lipschitz continuous with constants ${4(|\mathcal{A}_{i}|+|\mathcal{B}_{i}|)}$
for each block.
\end{proof}

\section{Convergence analysis}
\label{sec:conv_anal}
In Proposition \ref{prop1} we have shown that many, but not all, assumptions required
for the convergence analysis of the Block Coordinate Gradient Descent (BCGD)
algorithm in \cite{Beck_Tetruashvili_2011} are satisfied.
In what follows, we present an alternative and self-contained proof of
convergence of Algorithm~\ref{alg:alg}, which does not require Lipschitz
continuity of the gradients of the function $f$, as required in the proof of the
BCGD algorithm in \cite{Beck_Tetruashvili_2011}.
In our analysis we will make use of following results.

\begin{lemma}
\label{lem:descent}(\textbf{Descent Lemma}, Appendix A in \cite{Bertsekas-book-1995}) If ${g:\mathbb{R}^{p}\rightarrow \mathbb{R}}$ is a
continuously differentiable function whose gradient $\nabla g$ is Lipschitz
continuous with constant $L$ then%
\begin{equation}
g(y)\leq g(x)+\langle\nabla g(x),y -x \rangle+(L/2)\| x- y\|
^{2}~\text{for all}\,\, x,y\in \mathbb{R}^{p}.
\end{equation}

\end{lemma}

\begin{proposition}
\label{prop:prop3}For any $\boldsymbol{\mathbf{x}}^{0}\in \mathbb{R}^{dn}$ and for the
function $f$ of (\ref{problem}) the level set%
\begin{equation}
S=\{\boldsymbol{\mathbf{x}}\mid f(\boldsymbol{\mathbf{x}})\leq
f(\boldsymbol{\mathbf{x}}^{0})\}
\end{equation}
is bounded.
\end{proposition}

\begin{proof}
\noindent The function $f$ of (\ref{problem}) is coercive, i.e., for all $\boldsymbol{\mathbf{x}%
}$ such that $\|\boldsymbol{\mathbf{x}}\|\rightarrow\infty$ we
have $f(\boldsymbol{\mathbf{x}})\rightarrow\infty$. Thus, all level sets of $f$
are bounded.
\end{proof}

In the next proposition we show that any sequence generated by Algorithm
\ref{alg:alg} entails a non-increasing~sequence of function values
$\{f(\boldsymbol{\mathbf{x}}^{k})\}_{k=0}^{\infty}$.

\begin{proposition}
  \label{prop2} Let $\{x_{i}^{k}\}_{k=0}^{\infty}$ for all
  $i=1,2,\ldots,n,$ be sequences generated by Algorithm \ref{alg:alg}
  and let $\{f(\boldsymbol{\mathbf{x}}^{k})\}_{k=0}^{\infty}$ be the
  associated sequences of function values. Denote, for $0\leq i \leq
  n$,%
\begin{equation}
\boldsymbol{\mathbf{x}}^{k,i}:=(x_{1}^{k},x_{2}^{k},\ldots,x_{i}^{k}%
,x_{i+1}^{k-1},\ldots,x_{n}^{k-1}).
\end{equation}
Then, for every $k=1,2,\ldots,$ we have%
\begin{equation}
f(\boldsymbol{\mathbf{x}}^{k})\leq f(\boldsymbol{\mathbf{x}}^{k-1})
\label{eq27}%
\end{equation}
and%
\begin{equation}
\lim_{k\rightarrow\infty}\Vert\nabla_{i}f(\boldsymbol{\mathbf{x}}%
^{k,i-1})\Vert=0,\,\,\text{\thinspace for all}\,\,\,i=1,2,\ldots,n,
\label{eq5}%
\end{equation}
and%
\begin{equation}
\lim_{k\rightarrow\infty}\Vert\boldsymbol{\mathbf{x}}^{k,i}%
-\boldsymbol{\mathbf{x}}^{k,i-1}\Vert=0,\,\,\,\text{for all}\,\,\,i=1,2,\ldots
,n. \label{eq7}%
\end{equation}

\end{proposition}

\begin{proof}
\noindent Obviously, we have $\boldsymbol{\mathbf{x}}^{k}%
=\boldsymbol{\mathbf{x}}^{k,n}$ and $\boldsymbol{\mathbf{x}}^{k-1}%
=\boldsymbol{\mathbf{x}}^{k,0}$. In Proposition \ref{prop1}, we have shown that
the gradient of $f$ is block-coordinate-wise Lipschitz continuous with constants
$L_{i}=4(|\mathcal{A}_{i}|+|\mathcal{B}_{i}|)$ for all $i=1,2,\ldots,n.$ Thus,
by Lemma \ref{lem:descent}, we have for all $i=1,2,\ldots,n,$%
\begin{align}
f(\boldsymbol{\mathbf{x}}^{k,i})\leq &f(\boldsymbol{\mathbf{x}}^{k,i-1}%
)+\langle\nabla_{i}f(\boldsymbol{\mathbf{x}}^{k,i-1}),x_{i}^{k}-x_{i}%
^{k-1}\rangle\nonumber\\
&+(L_{i}/2)\| x_{i}^{k}-x_{i}^{k-1}\|^{2},
\label{eq30}%
\end{align}
where we used the equality%
\begin{equation}
\Vert \boldsymbol{\mathbf{x}}^{k,i}-\boldsymbol{\mathbf{x}}%
^{k,i-1}\Vert =\Vert x_{i}^{k}-x_{i}^{k-1}\Vert  \label{eql}%
\end{equation}
and
\begin{equation}
\langle \nabla f(\mathbf{x}^{k, i-1}), \mathbf{x}^{k, i} -
\mathbf{x}^{k, i-1} \rangle
=
\langle \nabla_i f(\mathbf{x}^{k, i-1}), x_i^k -
x_{i}^k \rangle .
\end{equation}
From the construction of the iterative sequence in Algorithm \ref{alg:alg} we
have, for all $i=1,2,\ldots,n$ and for all $k>0$,%
\begin{equation}
x_{i}^{k}-x_{i}^{k-1}=-(1/L_{i})\nabla_{i}f(\boldsymbol{\mathbf{x}}^{k,i-1}).
\label{eq3}%
\end{equation}
Plugging (\ref{eq3}) into (\ref{eq30}) yields%
\begin{equation}
f(\boldsymbol{\mathbf{x}}^{k,i-1})-f(\boldsymbol{\mathbf{x}}^{k,i}%
)\geq(1/2L_{i})\Vert\nabla_{i}f(\boldsymbol{\mathbf{x}}^{k,i-1})\Vert
^{2},\,\,\,i=1,\ldots,n. \label{eq4}%
\end{equation}
Summing over all inequalities and denoting ${\tau:=\max\{L_{i}\mid i=1,2,\ldots,n\}}$
we get%
\begin{equation}
f(\boldsymbol{\mathbf{x}}^{k-1})-f(\boldsymbol{\mathbf{x}}^{k})\geq
(1/2\tau)\sum_{i=1}^{n}\Vert\nabla_{i}f(\boldsymbol{\mathbf{x}}^{k,i-1}%
)\Vert^{2}, \label{eq41}%
\end{equation}
thus, proving (\ref{eq27}). From (\ref{eq41}) it follows that%
\begin{equation}
f(\boldsymbol{\mathbf{x}}^{0})\geq(1/2\tau)\sum_{k=0}^{\infty}\sum_{i=1}%
^{n}\Vert\nabla_{i}f(\boldsymbol{\mathbf{x}}^{k,i-1})\Vert^{2},
\end{equation}
which yields (\ref{eq5}), {since by interchanging the summation, an infinite sum of  positive values in the right-hand side is bounded from above by
  $f(\boldsymbol{\mathbf{x}}^{0})$}. Note that (\ref{eq3}) can be
equivalently rewritten in the form%
\begin{equation}
\left\Vert \boldsymbol{\mathbf{x}}^{k,i}-\boldsymbol{\mathbf{x}}%
^{k,i-1}\right\Vert =\left\Vert (1/L_{i})\nabla_{i}f(\boldsymbol{\mathbf{x}%
}^{k,i-1})\right\Vert \label{eq8}%
\end{equation}
which can be combined with (\ref{eq5}) to yield the desired result (\ref{eq7}).
\end{proof}

\begin{corollary}
\label{cor1} Any sequence $\{\boldsymbol{\mathbf{x}}^{k}\}_{k=0}^{\infty}$
generated by Algorithm \ref{alg:alg} is bounded.
\end{corollary}

\begin{proof}
\noindent From Proposition \ref{prop2}, it follows that the sequence
$\{\boldsymbol{\mathbf{x}}^{k}\}_{k=0}^{\infty}$ belongs to the bounded level
set ${S=\{\boldsymbol{\mathbf{x}}\mid f(\boldsymbol{\mathbf{x}})\leq
f(\boldsymbol{\mathbf{x}}^{0})\}}$, thus, by Proposition \ref{prop:prop3}, it
is bounded.
\end{proof}

Now we are ready to present and prove the main result.

\begin{theorem}
Let $\{\boldsymbol{\mathbf{x}}^{k}\}_{k=0}^{\infty}$ be any sequence generated
by Algorithm \ref{alg:alg}. Then every limit point of
$\{\boldsymbol{\mathbf{x}}^{k}\}_{k=0}^{\infty}$ is an optimal solution of the
problem (\ref{eq:minprob}) with the objective function (\ref{problem}) and the
sequence $\{f(\boldsymbol{\mathbf{x}}^{k})\}_{k=0}^{\infty}$ converges to the
optimal value $f^{\ast}$ of \eqref{eq:minprob}.
\end{theorem}

\begin{proof}
\noindent From Corollary \ref{cor1}, it follows that there exist converging
subsequence of the sequences $\{\boldsymbol{\mathbf{x}}^{k}\}_{k=0}^{\infty}$.
Let $\{\boldsymbol{\mathbf{x}}^{k_{j}}\}_{j=0}^{\infty}$ be such a convergent
subsequence and denote by $\boldsymbol{\mathbf{x}}^{\ast}$ its limit point. We
show that $\nabla f(\boldsymbol{\mathbf{x}}^{\ast})=0$, thus proving the
optimality of the point $\boldsymbol{\mathbf{x}}^{\ast}$.

Consider the sequence $\{\boldsymbol{\mathbf{x}}^{k_{j}}\}_{j=0}^{\infty}%
\cup\{\{\boldsymbol{\mathbf{x}}^{k_{j},i}\}_{i=1}^{n}\}_{j=0}^{\infty}$
obtained by combining these two sequences and denote it by
$\{\boldsymbol{\mathbf{y}}^{\ell}\}_{\ell=0}^{\infty}.$ From (\ref{eq7}) in
Proposition \ref{prop2}, it follows that the sequence $\{\boldsymbol{\mathbf{y}%
}^{\ell}\}_{\ell=0}^{\infty}$ also converges to the point
$\boldsymbol{\mathbf{x}}^{\ast}$. The sequence of gradients $\{\nabla
f(\boldsymbol{\mathbf{y}}^{\ell})\}_{\ell=0}^{\infty}$ then converges to
$\nabla f(\boldsymbol{\mathbf{x}}^{\ast})$ and so does every subsequence of it.

For the subsequence $\{\nabla f(\boldsymbol{\mathbf{x}}^{k_{j},1}%
)\}_{j=0}^{\infty}$, we conclude, from (\ref{eq5}) of Proposition \ref{prop2},
that it converges to a vector that has zeros in all its components that refer to
the part that comes from the partial derivatives with respect to components of
$x_{1}$. Thus, $\nabla f(\boldsymbol{\mathbf{x}}^{\ast})$ contains zeros in the
same components. Repeating this argument we reach the conclusion that all
elements of $\nabla f(\boldsymbol{\mathbf{x}}^{\ast})$ are zeros.\medskip

From (\ref{eq27}) it follows that the sequence $\{f(\boldsymbol{\mathbf{x}}%
^{k})\}_{k=0}^{\infty}$ is nonincreasing and, since it is also bounded from
below by $f^{\ast}$, it converges. Since we showed that any converging
subsequence $\{\boldsymbol{\mathbf{x}}^{k_{j}}\}_{j=0}^{\infty}$ converges to
the optimal solution, the subsequence $\{f(\boldsymbol{\mathbf{x}}^{k_{j}%
})\}_{j=0}^{\infty}$ converges to the optimal value $f^{\ast}$, thus the entire
sequence $\{f(\boldsymbol{\mathbf{x}}^{k})\}_{k=0}^{\infty}$ converges to
$f^{\ast}$.
\end{proof}

{
\begin{remark}
  It is noted that Algorithm~\ref{alg:alg} converges even if the
  intersection in \eqref{eq:problem} does not 
  contain the target nodes or even if it is empty--which could happen for negative
  measurement errors.  In fact, the proposed algorithm solves the
  optimization problem in \eqref{eq:minprob}, regardless of the extent
  of the intersection in~\eqref{alg:alg}.
\end{remark}}

\section{Numerical results}
\label{sec:sim_resul}
In this section, we evaluate the performance of Algorithm \ref{alg:alg} and two methods
proposed in \cite{Gholami_Eurasip_2011} and in \cite{Jia_Buehrer_2011} through computer simulations. To evaluate
the different methods, we consider a~2-dimensional network consisting of a number of
reference and target nodes. We will evaluate different algorithms ability to
localize target nodes under line-of-sight (LOS) and non-line-of-sight (NLOS) conditions.

The two algorithms of \cite{Gholami_Eurasip_2011} and
\cite{Jia_Buehrer_2011} as well as Algorithm \ref{alg:alg} cycle
through the networks in an ordered fashion. Although it is possible to
update the location of target nodes in an arbitrarily ordered fashion,
which may be more suitable for a practical implementation, algorithms
exchange the updated position estimates between nodes in the order
$1,2,\ldots,n$, where $n$ is the number of target nodes. In
\cite{Jia_Buehrer_2011} the authors used a method based on projections
(in parallel) onto spheres (projections onto the boundary of each
individual set). This method is sensitive to the choice of the initial
points and it may converge to local minima resulting in large
errors. To avoid converging to local minima the algorithm of
\cite{Jia_Buehrer_2011} is initialized at a target node with the
position of the closest reference node connected to that target.
However, in general, finding good initial estimates for all target
nodes is a challenging task in positioning problems. In
\cite{Gholami_Eurasip_2011}, a method based on projections onto convex
sets (POCS) was considered to localize target nodes.
{ Note that contrary to the algorithm proposed in this
  paper, there are no formal convergence proofs for the algorithms
  introduced in \cite{Gholami_Eurasip_2011,Jia_Buehrer_2011}.}

\subsection{Simulation setup}
The methods proposed in \cite{Gholami_Eurasip_2011,Jia_Buehrer_2011},
and Algorithm \ref{alg:alg}, are henceforth called cooperative
parallel projection onto boundary (Coop.~PPB), cooperative projection
onto convex sets (Coop.~ POCS), and cooperative parallel projection
method (Coop.~PPM), respectively. We have conducted computer
simulations for different scenarios and evaluated the algorithms based
on two metrics: convergence and accuracy. We have considered a
$100\times100~\mathrm{m}^{2}$ square area shown in
Fig.\,\ref{fig:network} with a number of reference nodes at fixed
positions. In the simulation, we randomly distributed a number of
target nodes inside the area\footnote{We have also assessed the
  algorithms for 3-dimensional networks. The results show similar
  behavior as for 2-dimensional networks.}.

\begin{figure}
\centering
\psfrag{xlabel}[cc][][.8]{$x$-axis [m]}
\psfrag{ylabel}[cc][][.8]{$y$-axis [m]}
\psfrag{a1}[cc][][.8]{$a_1$}
\psfrag{a2}[cc][][.8]{$a_2$}
\psfrag{a3}[cc][][.8]{$a_3$}
\psfrag{a4}[cc][][.8]{$a_4$}
\psfrag{a5}[cc][][.8]{$a_5$}
\includegraphics[width=80mm]{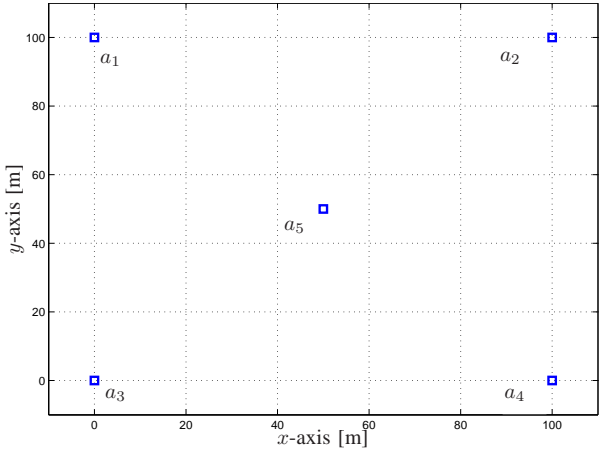}
\caption{Layout of the network considered in simulations, where
  squares denote the positions of reference nodes and target nodes are
  randomly distributed inside the convex hull of the reference nodes.}
\label{fig:network}
\end{figure}

The measurement errors in \eqref{eq:dis_estimate} are modeled as
\[
\varepsilon_{mn} =
\begin{cases}
  \varepsilon_{G,mn}, & \text{LOS conditions},\\
  \varepsilon_{G,mn} + b_{mn} \varepsilon_{U,mn}, & \text{NLOS conditions},
\end{cases}
\]
where $\varepsilon_{G,mn}$ are iid zero-mean Gaussian random variables
with variance $\sigma^2$, $b_{mn}\in\{0, 1\}$ are iid Bernoulli random
variables with parameter $p_{\text{NLOS}} = \Pr\{b_{mn} = 1\}$, and
$\varepsilon_{U,mn}\in [0, L]$ are iid uniformly distributed random
variables.  The validity of uniform distribution to model NLOS noise
has been considered in various works, e.g., \cite{Mats_thesis_2008,Urruela_2006,Gholami_Eurasip_2011}. In our
simulations, we used $\sigma = 1$~m, $p_{\text{NLOS}} = 0.2$, and $L = 20$~m.
The connectivity is defined based on the actual distances between
sensor nodes. Namely, if the distance between two nodes is $40$~m or
less, we assume that the nodes are connected.  To assess the
positioning algorithms, we consider the cumulative distribution
function (CDF) of the position error for all
target nodes. That is, we define the position error for target node
$i$ as
\begin{align}
E_{i,m}=\|x^{K}_{i,m}-x_{i,m}\|,\quad i=1,2,\ldots,n, \quad m = 1, 2,
\ldots, N
\end{align}
where $x^{K}_{i,m}$ is an estimate of target node $i$ position,
$x_{i,m}$, after $K$ iterations for the $m$th network realization, and $N$ is the
total number of network realizations. In the results below, we have used $K = 300$, which is large enough to ensure convergence
and then compute the empirical CDF of the position error for all target nodes collected
in $\mathbf{E}$ as
\begin{align}
\text{CDF}(\alpha)&=\mathrm{Pr}(\text{position error}\leq\alpha)\nonumber\\
&=\frac{\sum_{i=1}^n\sum_{m=1}^{N}I(E_{i,m}-\alpha)}{nN},
\end{align}
where the indicator function $I(t)$ is defined as
\begin{align}
I(t)=\left
\{\begin{array}[c]{ll}%
1,& \text{if}~t\geq0,\\
0,&\text{otherwise.}
\end{array}\right.
\end{align}

To study the convergence rate, we consider the average residual defined as~
\begin{align}
\bar{r}_k=1/(nN)\sum_{m=1}^N\|\mathbf{x}^{k}_m-\mathbf{x}^{k-1}_m\|,
\end{align}
where
\begin{align}
\mathbf{x}^{k}_m=\left({x}_{1,m}^{k},{x}_{2,m}^{k}\ldots,{x}_{n,m}^{k}\right).
\end{align}


To initialize the algorithm of \cite{Jia_Buehrer_2011}, we set the initial
estimate for a target node as the position of the closest reference node
connected to the target. For a target $i$ which is not connected to any
reference node, we pick the position of the closest target (already
initialized) connected to the target node~$i$.

\subsection{Performance of algorithms}

In this section, we assess the above mentioned algorithms for LOS and
NLOS scenarios.  In Fig.~\ref{fig:CDF_4ref_NLOS} and
Fig.~\ref{fig:CDF_5ref_NLOS}, we plot the CDF of the position error of
the algorithms for various numbers of reference and target nodes in
the LOS scenario. The figures show a superior performance for the
proposed approach Coop.~PPM compared with the others, especially with
respect to Coop.~PPB.  Coop.~PPB has poor performance for both small
number of reference nodes and large number of target nodes.  The
reason is that finding good initial points is difficult when the
number of reference nodes decreases or the number target nodes
increases. From the plots, we observe that Coop.~POCS shows good
performance compared with Coop.~PPB.

\begin{figure*}
\subfigure[]{
\psfrag{xlabel}[cc][][.8]{Position error [m]}
\psfrag{ylabel}[cc][][.8]{CDF}
\psfrag{Coop. Proj. onto Boundary}[cc][][.5]{Coop. PPB\qquad\qquad\quad}
\psfrag{Coop. POCS}[cc][][.5]{Coop.POCS}
\psfrag{Coop. PPM}[cc][][.5]{Coop. PPM}
\label{fig:cdf_4ref_20tar}
\includegraphics[width=80mm]{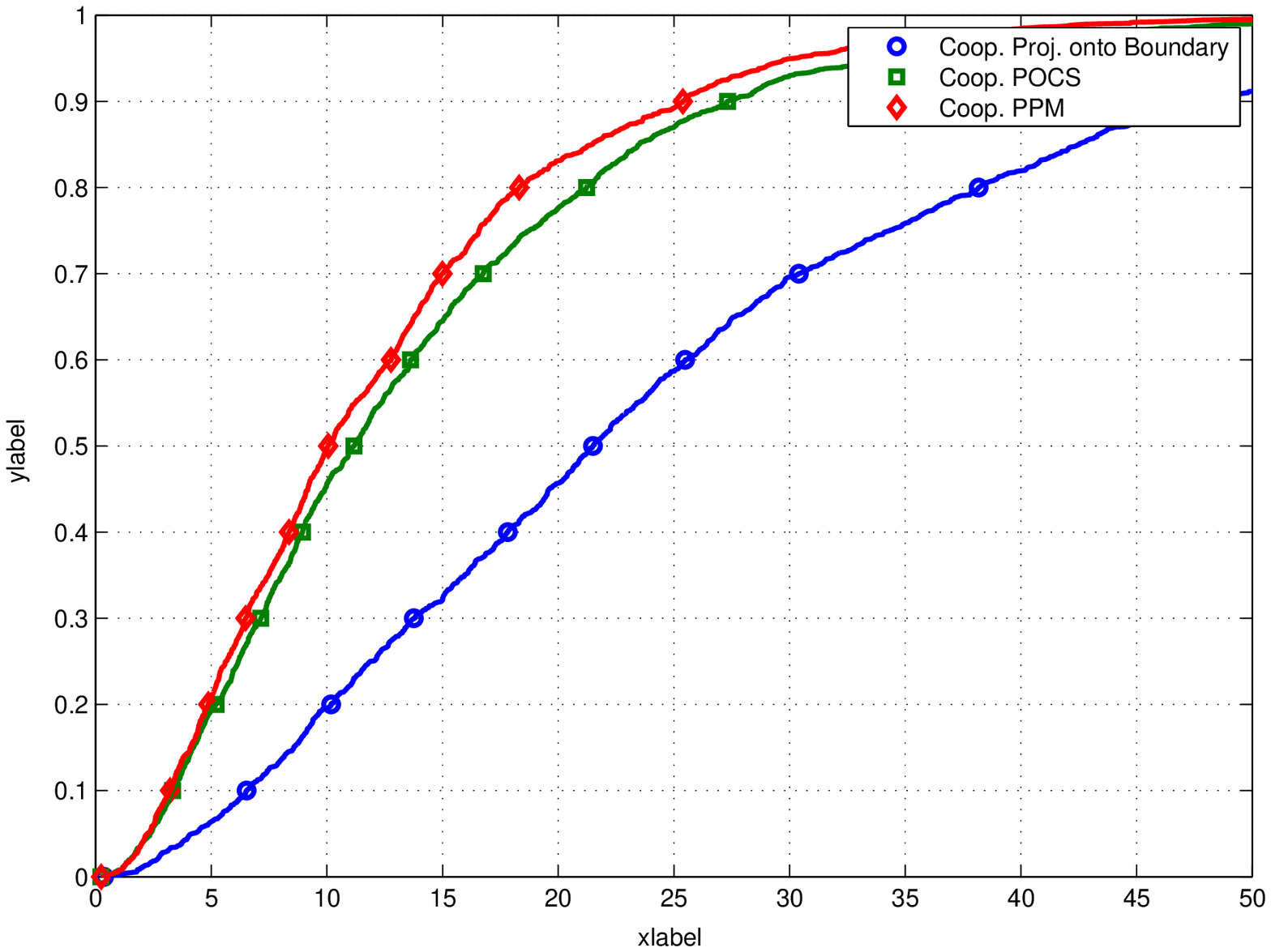}} \subfigure[]{
\psfrag{xlabel}[cc][][.8]{Position error [m]}
\psfrag{ylabel}[cc][][.8]{CDF}
\psfrag{Coop. Proj. onto Boundary}[cc][][.5]{Coop. PPB\qquad\qquad\quad}
\psfrag{Coop. POCS}[cc][][.5]{Coop.POCS}
\psfrag{Coop. PPM}[cc][][.5]{Coop. PPM}
\label{fig:cdf_4ref_30tar}
\includegraphics[width=80mm]{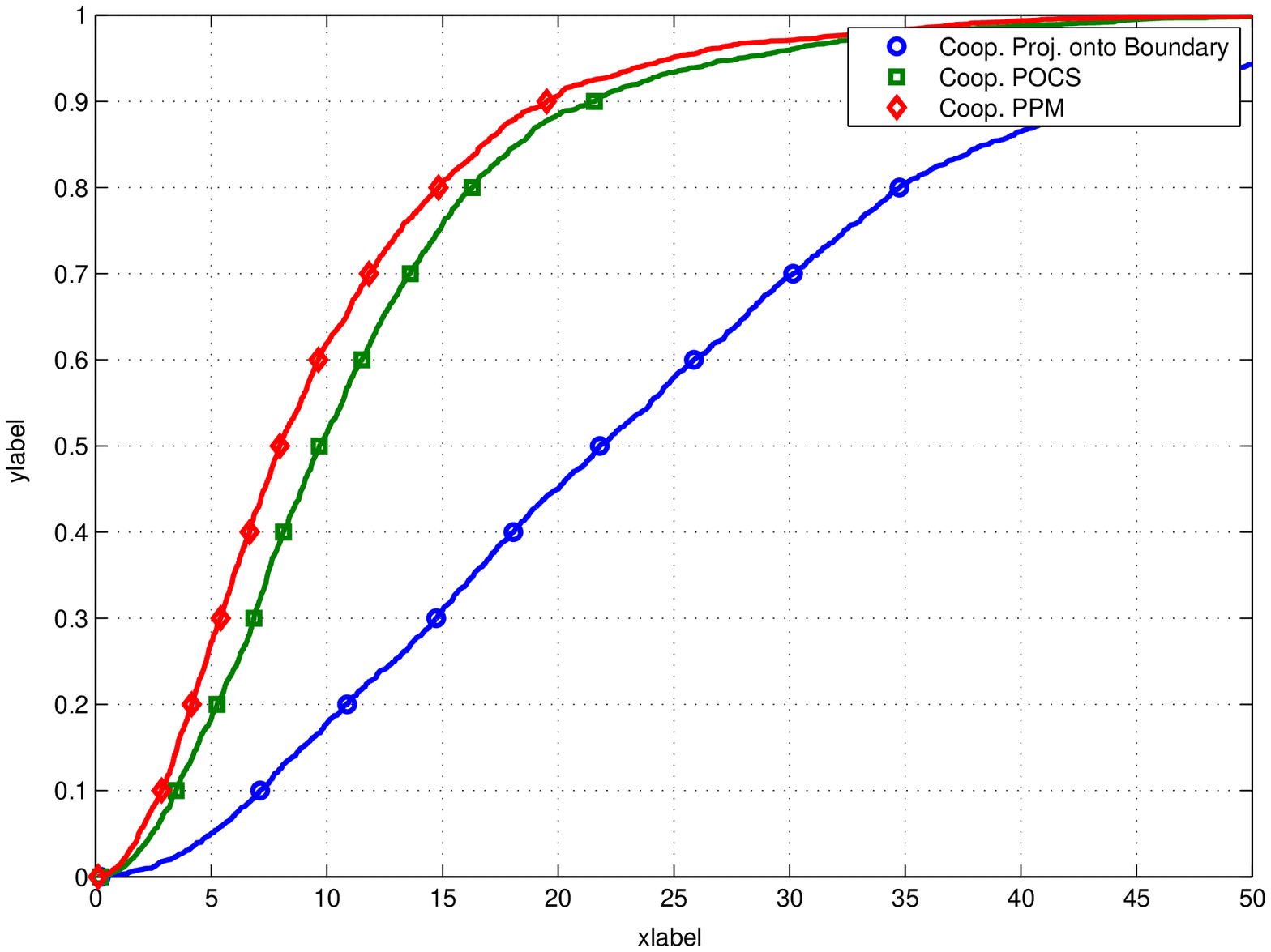}}\newline%
\subfigure[]{
\psfrag{xlabel}[cc][][.8]{Position error [m]}
\psfrag{ylabel}[cc][][.8]{CDF}
\psfrag{Coop. Proj. onto Boundary}[cc][][.5]{Coop. PPB\qquad\qquad\quad}
\psfrag{Coop. POCS}[cc][][.5]{Coop.POCS}
\psfrag{Coop. PPM}[cc][][.5]{Coop. PPM}
\label{fig:cdf_4ref_60tar}
\includegraphics[width=80mm]{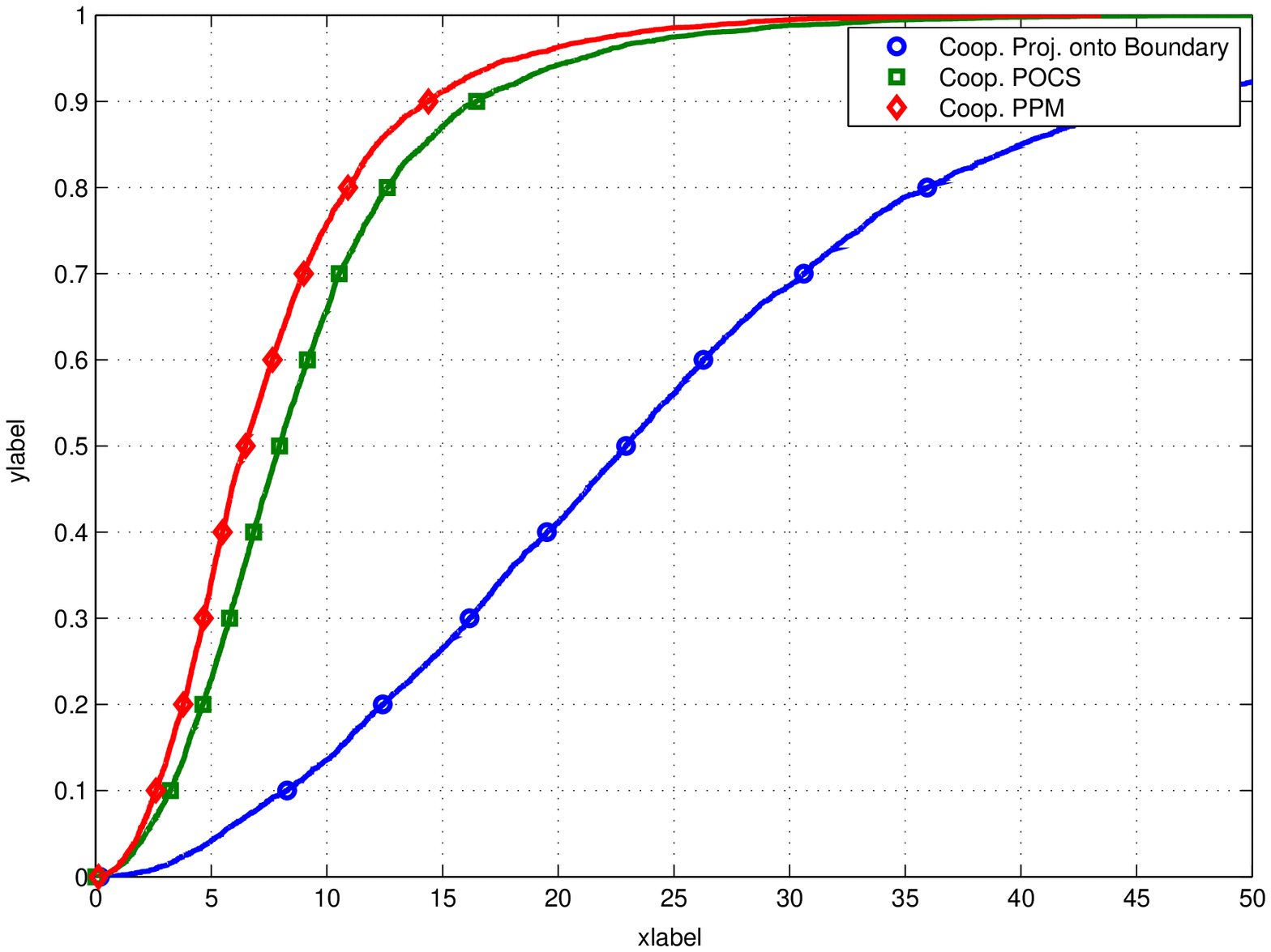}} \subfigure[]{
\psfrag{xlabel}[cc][][.8]{Position error [m]}
\psfrag{ylabel}[cc][][.8]{CDF}
\psfrag{Coop. Proj. onto Boundary}[cc][][.5]{Coop. PPB\qquad\qquad\quad}
\psfrag{Coop. POCS}[cc][][.5]{Coop.POCS}
\psfrag{Coop. PPM}[cc][][.5]{Coop. PPM}
\label{fig:cdf_4ref_100tar}
\includegraphics[width=80mm]{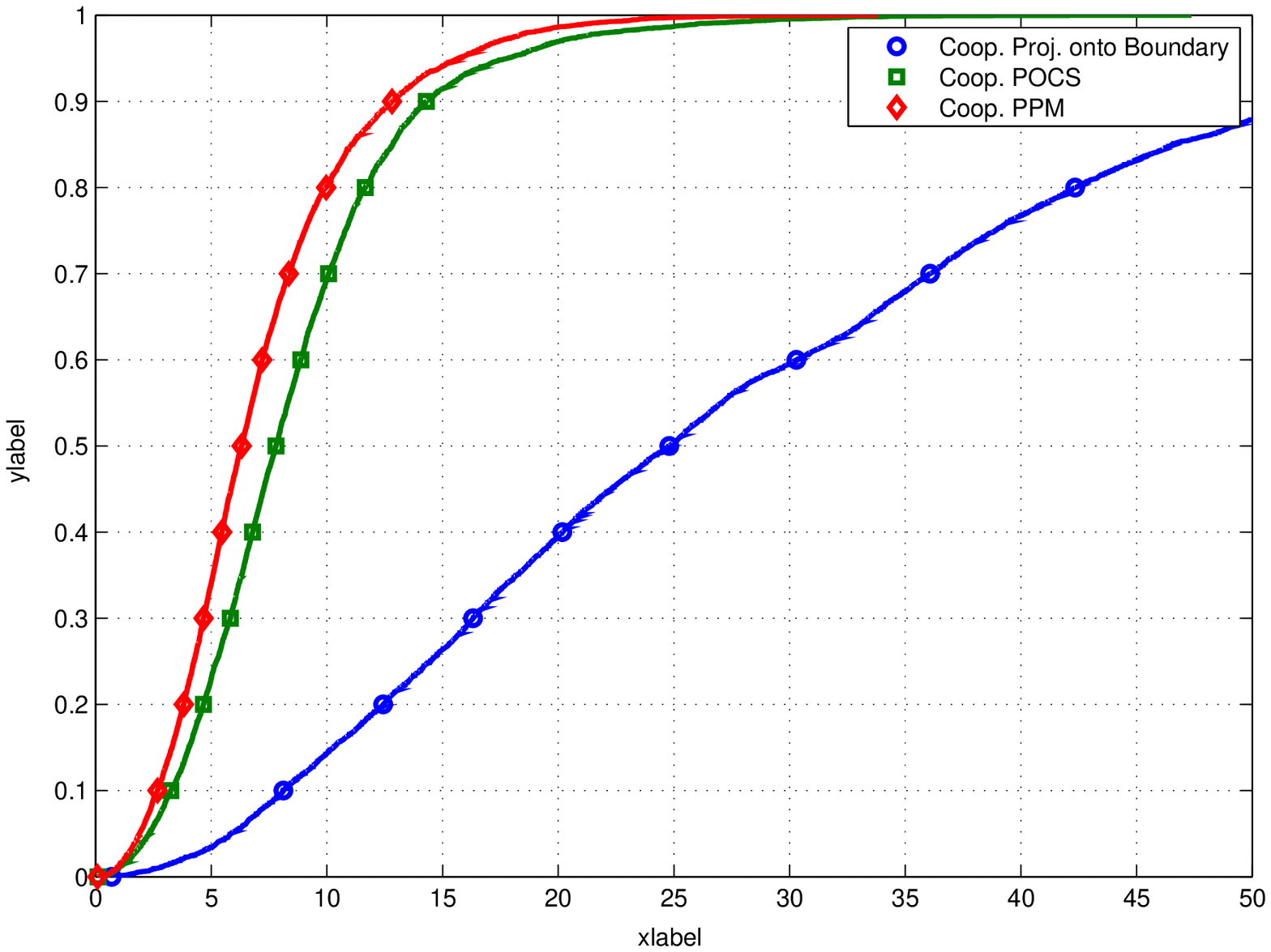}}
\caption{The CDF
of the position error of different algorithms in LOS scenarios for four reference nodes $a_1, a_2, a_3, a_4$ and
\subref{fig:cdf_4ref_20tar} 20 target nodes,
\subref{fig:cdf_4ref_30tar} 30 target nodes,
\subref{fig:cdf_4ref_60tar} 60 target nodes, and
\subref{fig:cdf_4ref_100tar} 100 target nodes.}%
\label{fig:CDF_4ref_NLOS}%
\end{figure*}

\begin{figure*}
\subfigure[]{
\psfrag{xlabel}[cc][][.8]{Position error [m]}
\psfrag{ylabel}[cc][][.8]{CDF}
\psfrag{Coop. Proj. onto Boundary}[cc][][.5]{Coop. PPB\qquad\qquad\quad}
\psfrag{Coop. POCS}[cc][][.5]{Coop.POCS}
\psfrag{Coop. PPM}[cc][][.5]{Coop. PPM}
\label{fig:cdf_5ref_20tar}
\includegraphics[width=80mm]{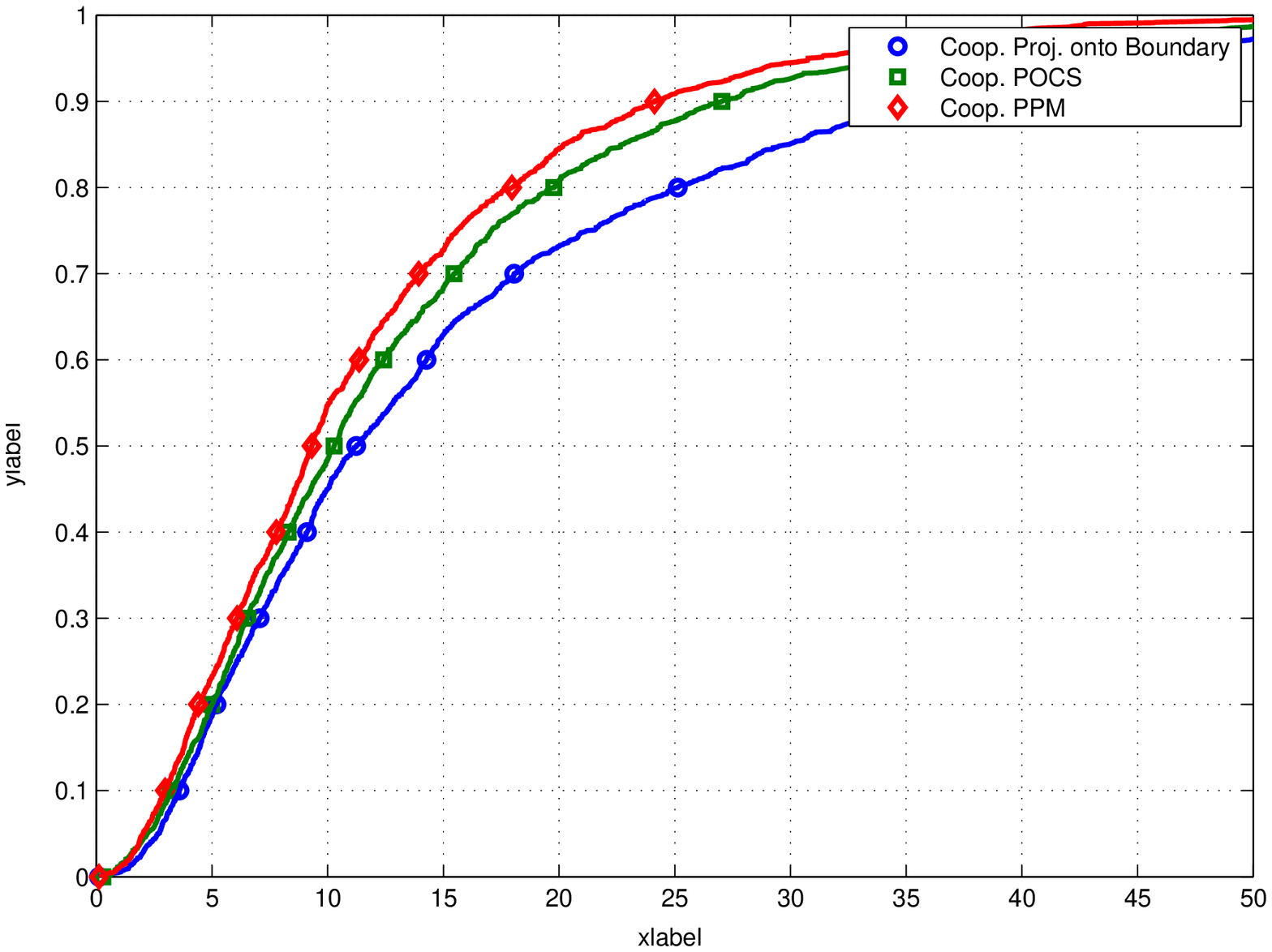}} \subfigure[]{
\psfrag{xlabel}[cc][][.8]{Position error [m]}
\psfrag{ylabel}[cc][][.8]{CDF}
\psfrag{Coop. Proj. onto Boundary}[cc][][.5]{Coop. PPB\qquad\qquad\quad}
\psfrag{Coop. POCS}[cc][][.5]{Coop.POCS}
\psfrag{Coop. PPM}[cc][][.5]{Coop. PPM}
\label{fig:cdf_5ref_30tar}
\includegraphics[width=80mm]{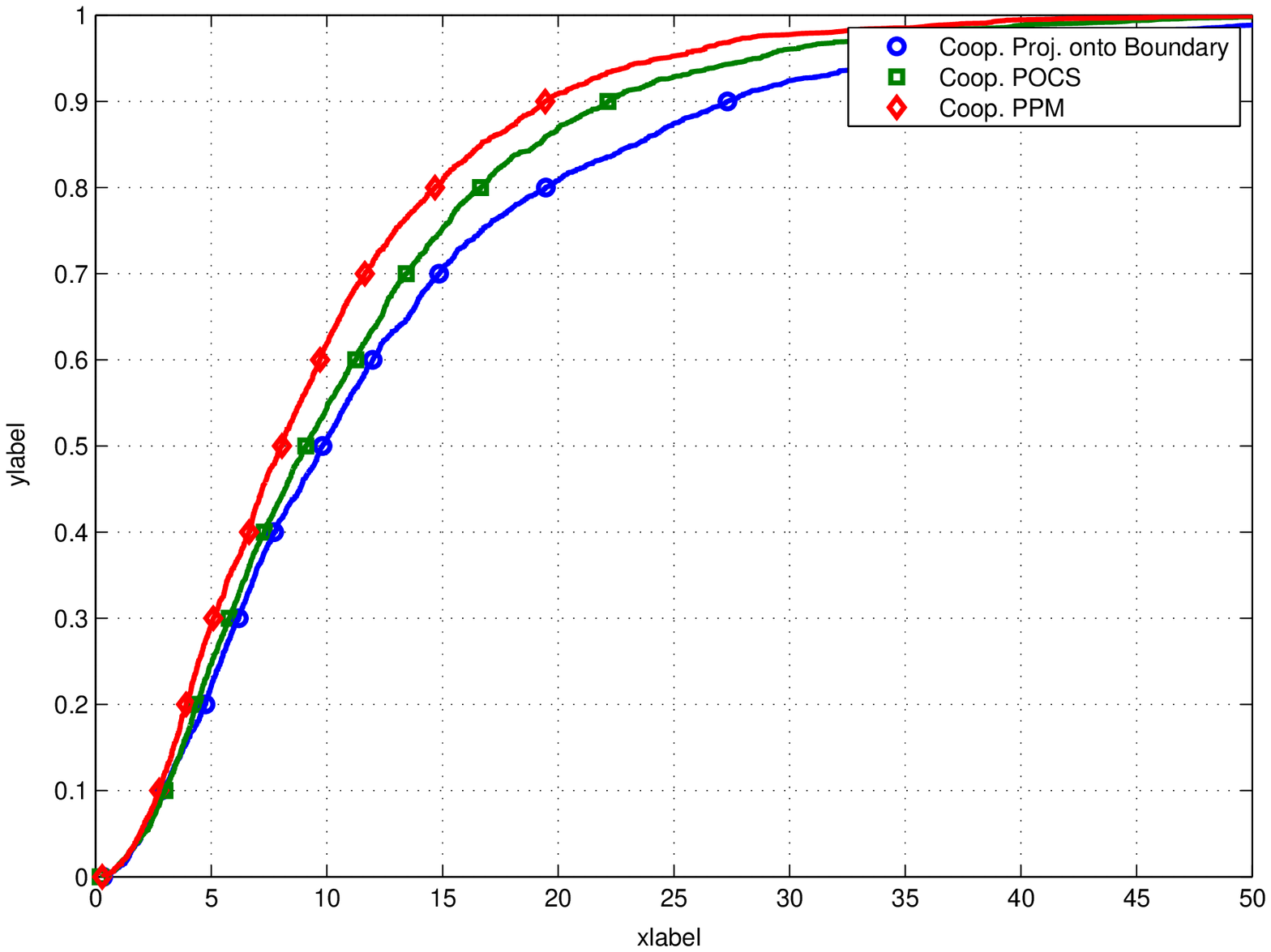}}\newline%
\subfigure[]{
\psfrag{xlabel}[cc][][.8]{Position error [m]}
\psfrag{ylabel}[cc][][.8]{CDF}
\psfrag{Coop. Proj. onto Boundary}[cc][][.5]{Coop. PPB\qquad\qquad\quad}
\psfrag{Coop. POCS}[cc][][.5]{Coop.POCS}
\psfrag{Coop. PPM}[cc][][.5]{Coop. PPM}
\label{fig:cdf_5ref_60tar}
\includegraphics[width=80mm]{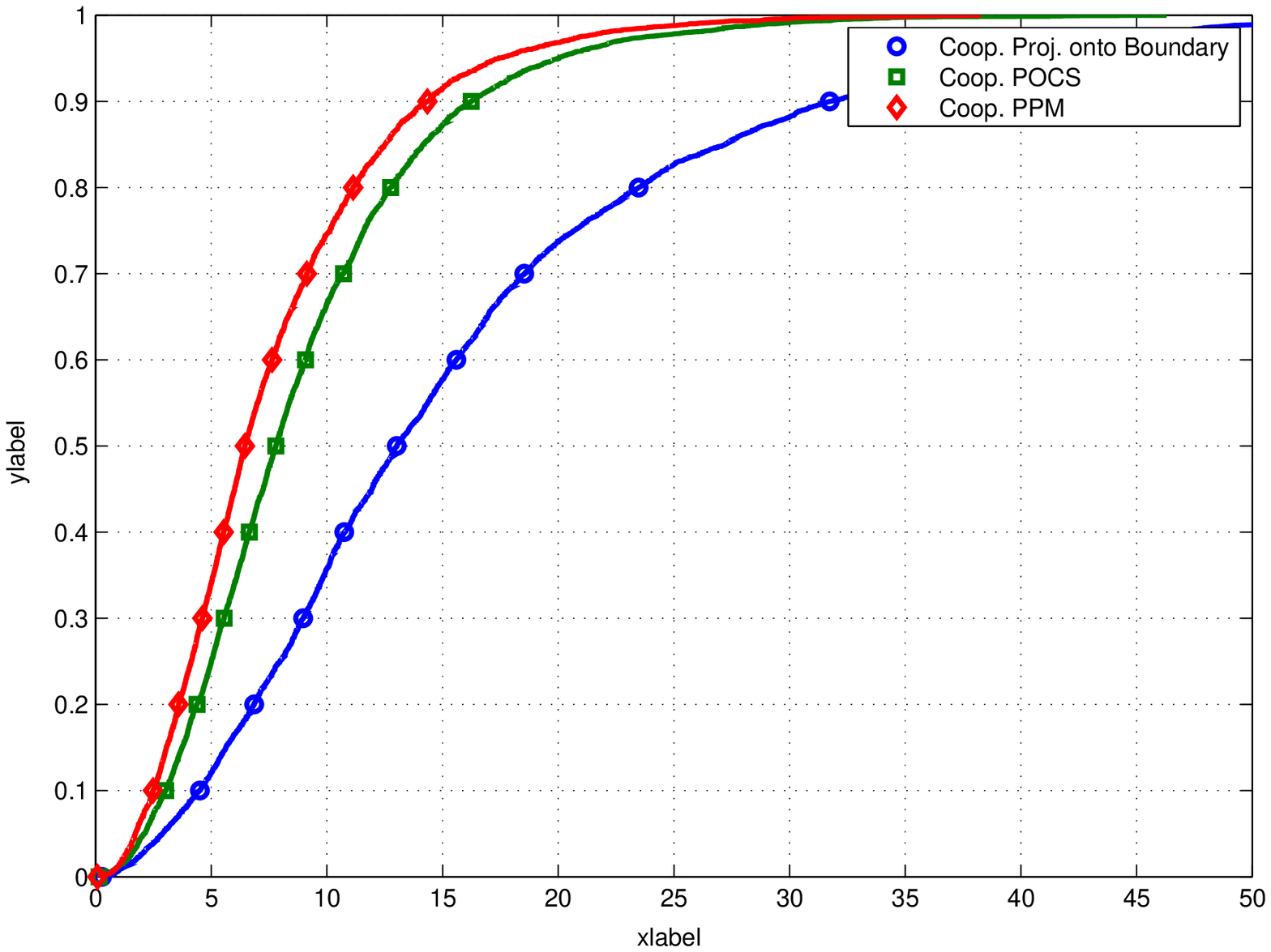}} \subfigure[]{
\psfrag{xlabel}[cc][][.8]{Position error [m]}
\psfrag{ylabel}[cc][][.8]{CDF}
\psfrag{Coop. Proj. onto Boundary}[cc][][.5]{Coop. PPB\qquad\qquad\quad}
\psfrag{Coop. POCS}[cc][][.5]{Coop.POCS}
\psfrag{Coop. PPM}[cc][][.5]{Coop. PPM}
\label{fig:cdf_5ref_100tar}
\includegraphics[width=80mm]{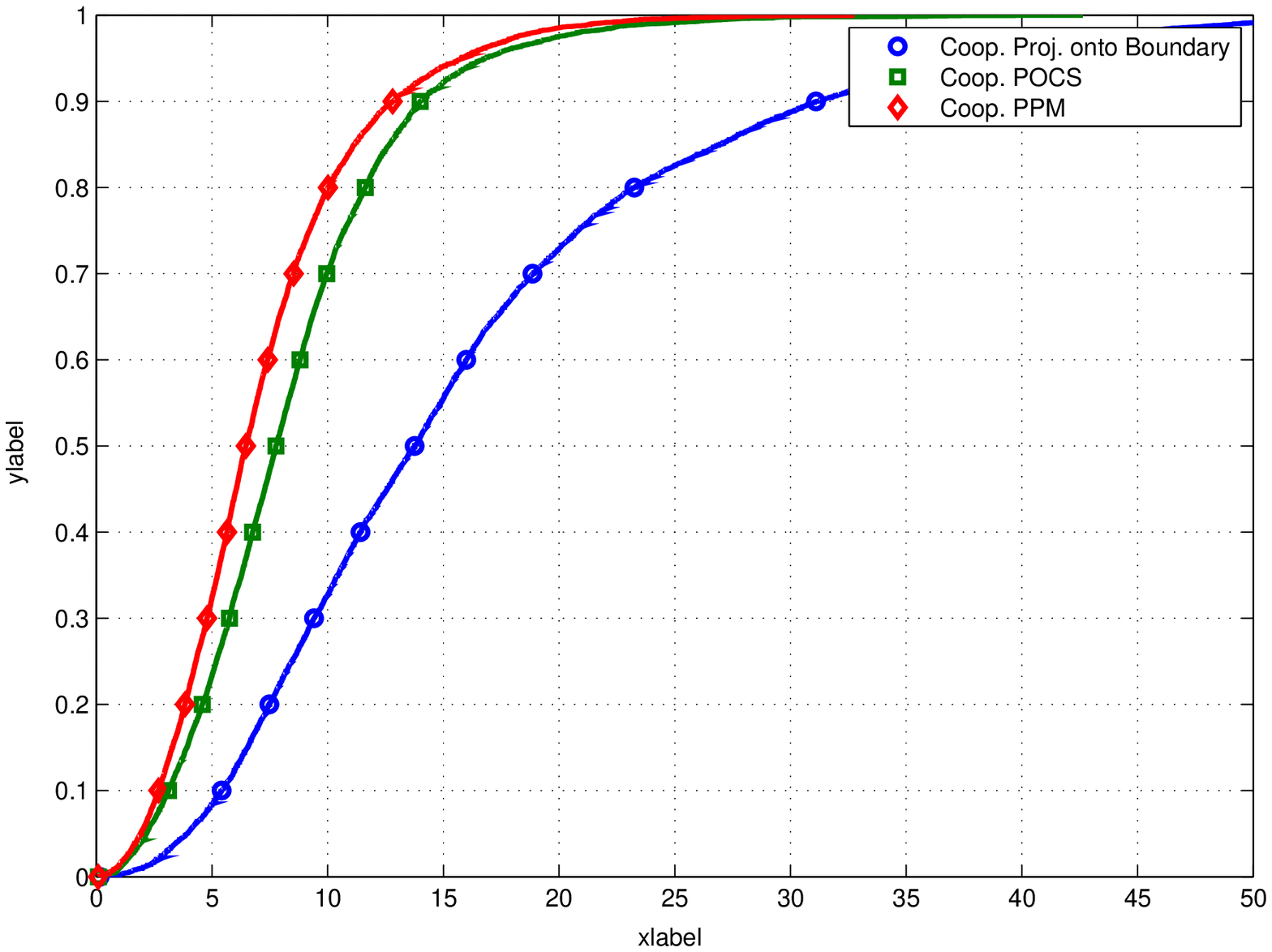}}
\caption{The CDF
of the position error of different algorithms in LOS scenario for five reference nodes $a_1, a_2, a_3, a_4, a_5$ and
\subref{fig:cdf_5ref_20tar} 20 target nodes,
\subref{fig:cdf_5ref_30tar} 30 target nodes,
\subref{fig:cdf_5ref_60tar} 60 target nodes, and
\subref{fig:cdf_5ref_100tar} 100 target nodes.}%
\label{fig:CDF_5ref_NLOS}%
\end{figure*}

To further investigate the algorithms, we assume 20\% of distance
measurements are NLOS.  The position error CDFs are plotted in
Fig.~\ref{fig:NLOS} for different numbers of target nodes in NLOS
conditions.  In this simulation, we used four reference nodes $a_1,
\ldots, a_4$.  From the plots, we observe that Coop.~POCS and
Coop.~PPM show better performance compared with Coop.~PPB and that
they are robust against NLOS measurements. The poor performance of
Coop.~PPB in the considered scenario, is mainly due to poor
initialization.

\begin{figure*}[ptb]
\subfigure[]{
\psfrag{xlabel}[cc][][.8]{Position error [m]}
\psfrag{ylabel}[cc][][.8]{CDF}
\psfrag{Coop. Proj. onto Boundary}[cc][][.5]{Coop.~PPB\qquad\qquad\quad}
\psfrag{Coop. POCS}[cc][][.5]{Coop.~POCS}
\psfrag{Coop. PPM}[cc][][.5]{Coop. PPM}
\label{fig:cdf_4ref_30tar_nlos}
\includegraphics[width=80mm]{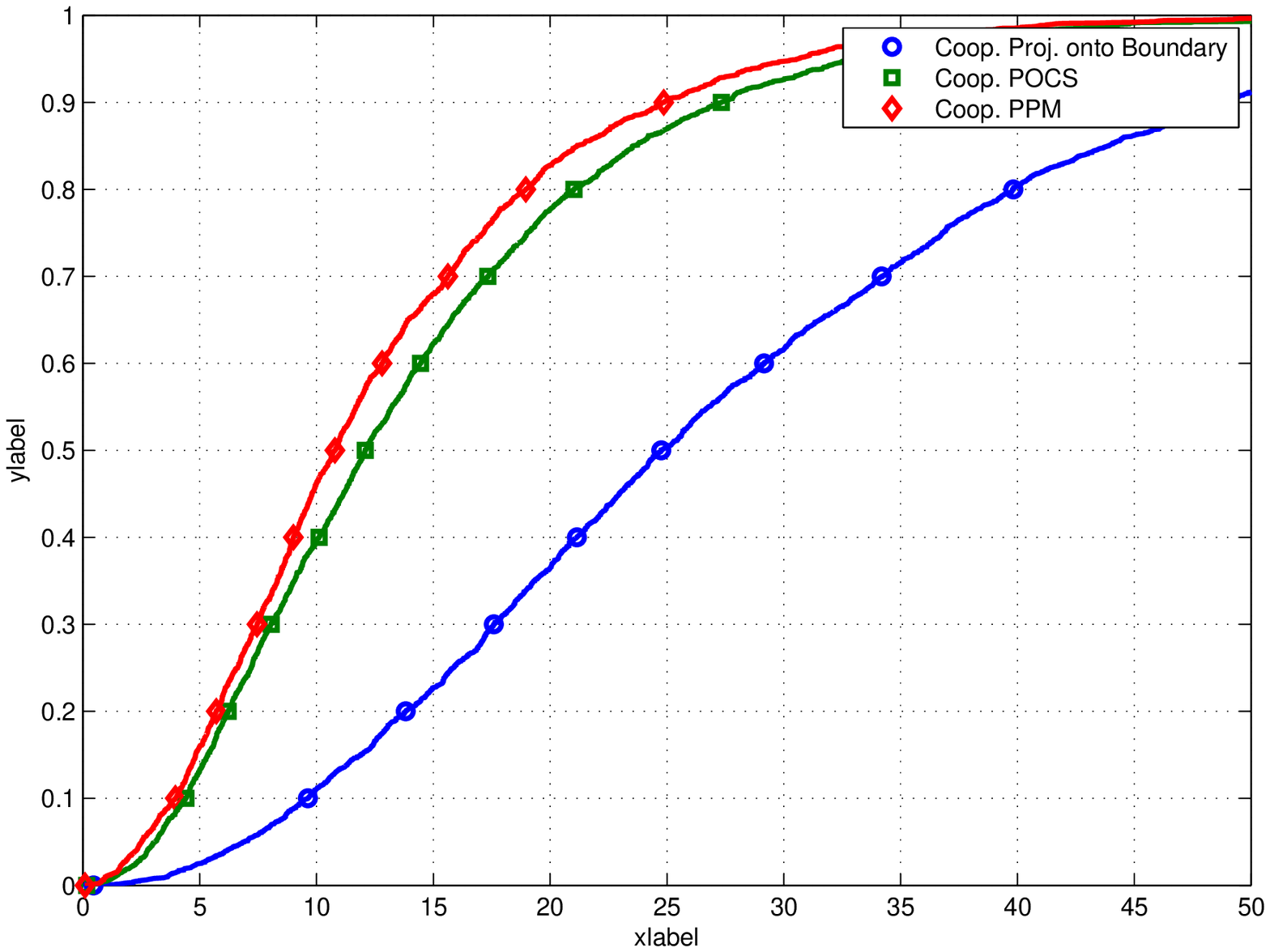}} \subfigure[]{
\psfrag{xlabel}[cc][][.8]{Position error [m]}
\psfrag{ylabel}[cc][][.8]{CDF}
\psfrag{Coop. Proj. onto Boundary}[cc][][.5]{Coop.~PPB\qquad\qquad\quad}
\psfrag{Coop. POCS}[cc][][.5]{Coop.~POCS}
\psfrag{Coop. PPM}[cc][][.5]{Coop. PPM}
\label{fig:cdf_4ref_60tar_nlos}
\includegraphics[width=80mm]{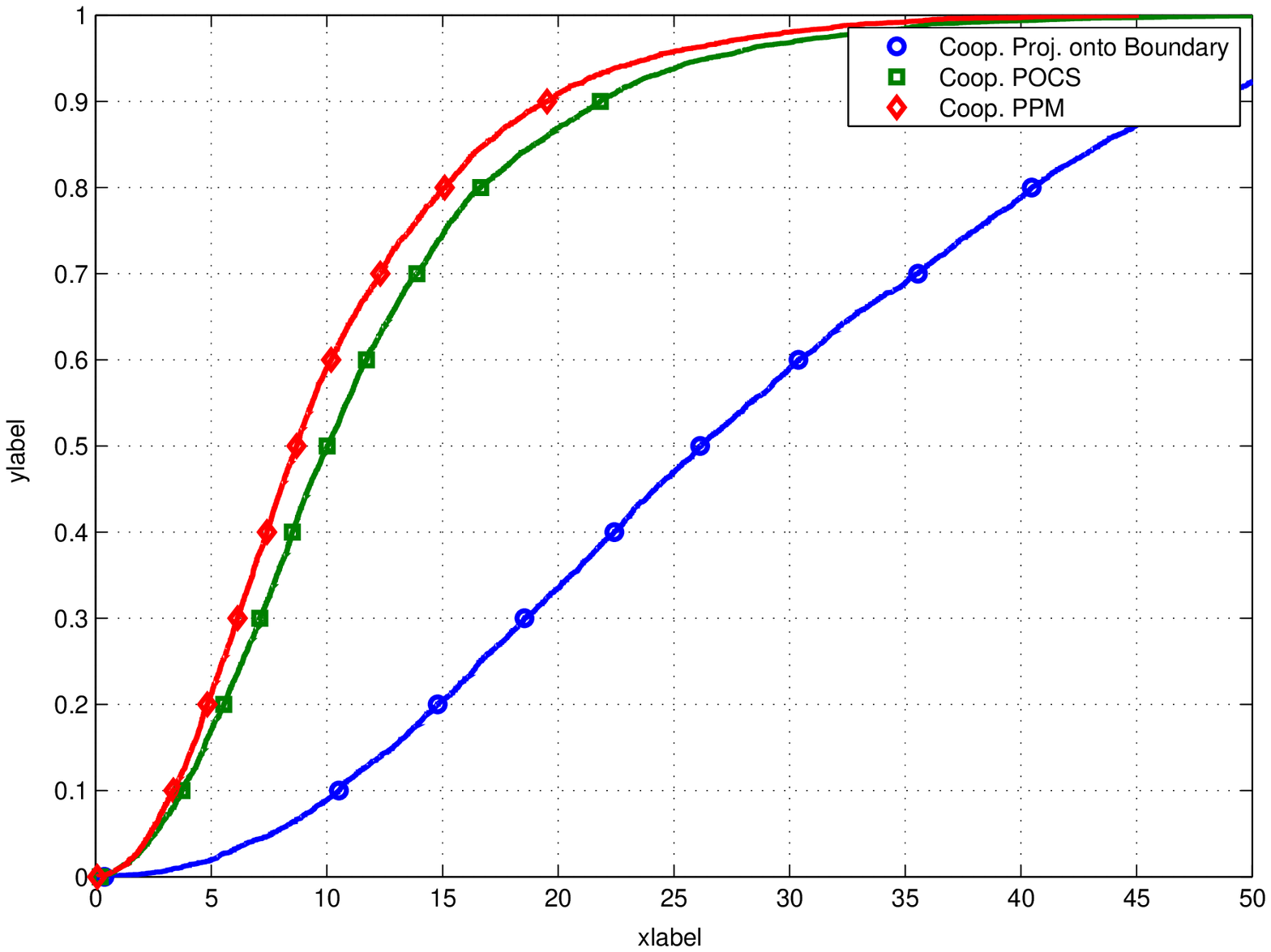}} \caption{The
CDF of the position error in NLOS scenario for four reference nodes and \subref{fig:cdf_4ref_30tar_nlos} 30 target nodes and \subref{fig:cdf_4ref_60tar_nlos} 60 target nodes.}%
\label{fig:NLOS}%
\end{figure*}

Based on the results presented here and on additional simulation
results that we performed, we can conclude that the method proposed in
this paper has superior performance compared with other projection
approaches available in the positioning literature for sparse networks
in which target nodes are connected to a few other nodes.

\subsection{Convergence speed}
\begin{figure*}
\subfigure[]{
\psfrag{xlabel}[cc][][.8]{iterations, $k$}
\psfrag{ylabel}[cc][][1]{$\bar{r}_k$}
\psfrag{Coop.PPB}[cc][][.45]{Coop.~PPB}
\psfrag{Coop.POCS}[cc][][.45]{Coop.~POCS}
\psfrag{Coop.PPM}[cc][][.45]{Coop.~PPM}
\label{fig:5ref_20tar_LOS}
\includegraphics[width=80mm]{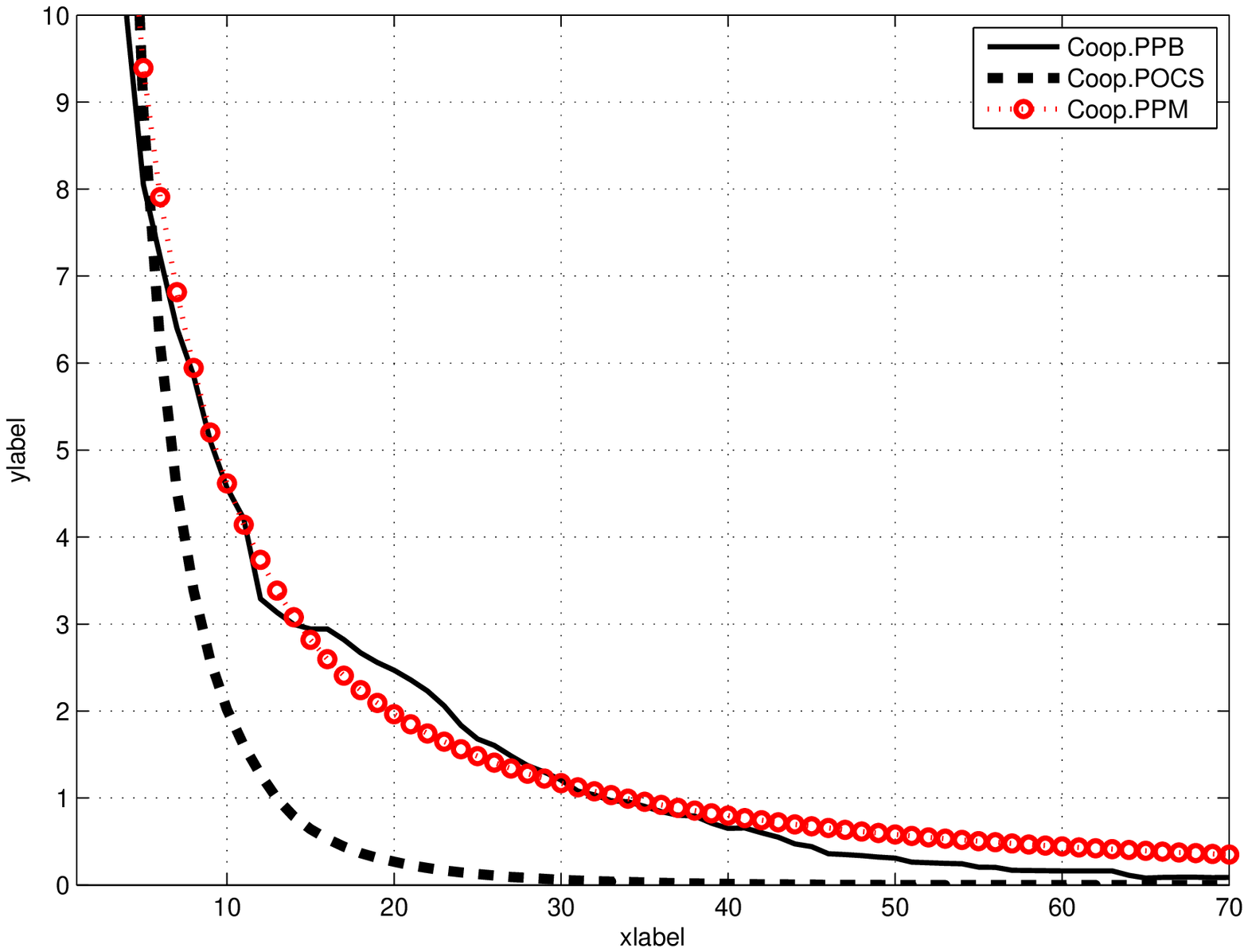}}
\subfigure[]{
\psfrag{xlabel}[cc][][.8]{iterations, $k$}
\psfrag{ylabel}[cc][][1]{$\bar{r}_k$}
\psfrag{Coop.PPB}[cc][][.45]{Coop.~PPB}
\psfrag{Coop.POCS}[cc][][.45]{Coop.~POCS}
\psfrag{Coop.PPM}[cc][][.45]{Coop.~PPM}
\label{fig:5ref_30tar_LOS}
\includegraphics[width=80mm]{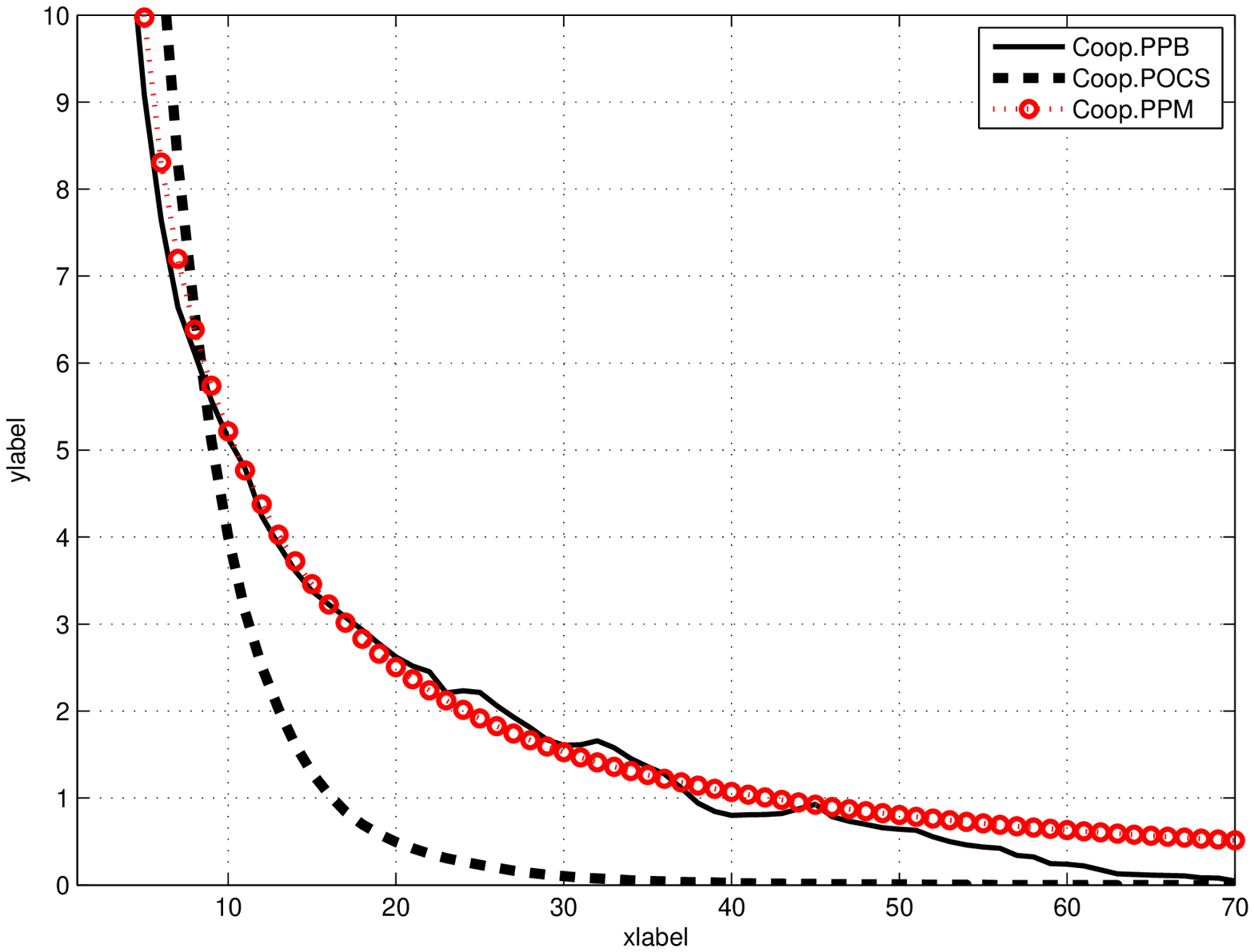}}\\
\subfigure[]{
\psfrag{xlabel}[cc][][.8]{iterations, $k$}
\psfrag{ylabel}[cc][][1]{$\bar{r}_k$}
\psfrag{Coop.PPB}[cc][][.45]{Coop.~PPB}
\psfrag{Coop.POCS}[cc][][.45]{Coop.~POCS}
\psfrag{Coop.PPM}[cc][][.45]{Coop.~PPM}
\label{fig:5ref_60tar_LOS}
\includegraphics[width=80mm]{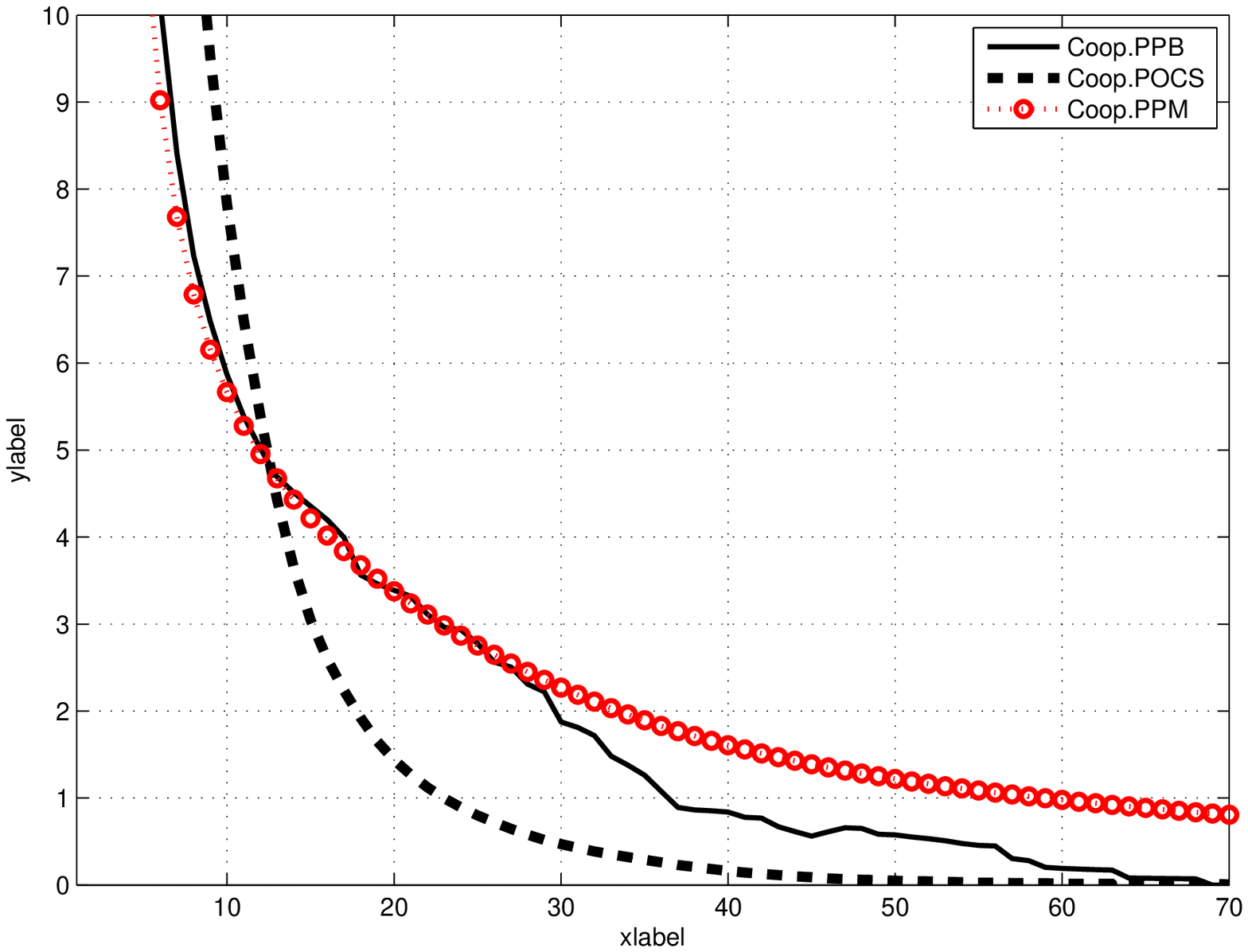}}
\subfigure[]{
\psfrag{xlabel}[cc][][.8]{iterations, $k$}
\psfrag{ylabel}[cc][][1]{$\bar{r}_k$}
\psfrag{Coop.PPB}[cc][][.45]{Coop.~PPB}
\psfrag{Coop.POCS}[cc][][.45]{Coop.~POCS}
\psfrag{Coop.PPM}[cc][][.45]{Coop.~PPM}
\label{fig:5ref_100tar_LOS}
\includegraphics[width=80mm]{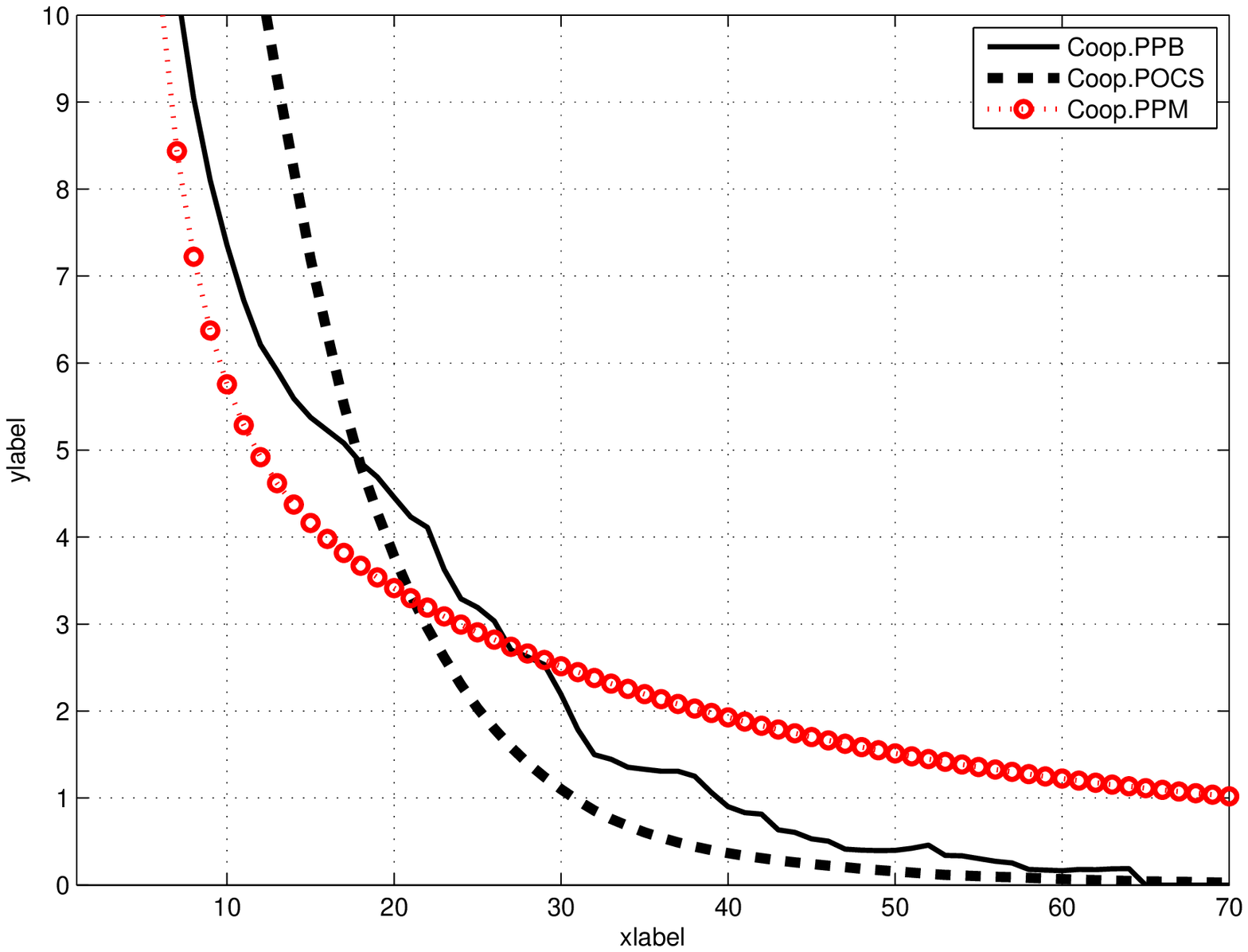}}
\caption{Convergence speed of average residuals $\bar{r}_k=1/(nN)\sum_{m=1}^N\|\mathbf{x}^{k}_m-\mathbf{x}^{k-1}_m\|$ versus the number of
iterations, $k$, for different algorithms in LOS conditions for five reference nodes and \subref{fig:5ref_20tar_LOS} 20 target nodes,
\subref{fig:5ref_30tar_LOS} 30 target nodes,
\subref{fig:5ref_60tar_LOS} 60 target nodes, and \subref{fig:5ref_100tar_LOS} 100 target nodes.}%
\label{fig:conv_rate}%
\end{figure*}

In this section, we evaluate the convergence speed of the
above mentioned algorithms through simulations. We consider the LOS
scenario for the network shown in Fig.~\ref{fig:network}.  The
convergence speed is plotted in Fig.\,\ref{fig:conv_rate} for
different numbers of target nodes. For Coop.~POCS, we use $5D_i$
iterations for locally updating target node $i$, where $D_i$ is the
number of sensor nodes (reference or localized target) with known or
estimated location connected to target $i$. That is, for target $i$, we
continue sequential projection onto $D_i$ balls corresponding to the nodes
connected to target $i$. We first set the relaxation parameters equal
to one and then decrease them to the value used
in~\cite{Gholami_Eurasip_2011} after $3D_i$. In the figure, the
$x$-axis shows the number of iterations for updating the vector
$\mathbf{x}^{k}=\left({x}_{1}^k,x_{2}^k,\ldots,{x}_{n}^{k}\right)$.
It is observed that all algorithms converge after a few iterations.
The Coop.~POCS uses more local updatings, and it shows faster global
convergence compared to the two other methods.
Finally, from this figure we also see that the convergence rate of
Coop.~PPB may not be monotone, while Coop.~PPM and Coop.~POCS show
monotonic convergence.  {In fact, the Coop.~PPB,
  objective function is nonconvex, and we need to ensure that the
  starting point is sufficently close to the global solution, since
  the algorithm might otherwise converge to a local minimum. We also
  note that the proposed algorithm has relatively slow convergence
  after, say, 40--50 iterations. However, monotonic convergence is
  guaranteed by Theorem~9.}

\section{Conclusions}
In this paper, we have considered a geometric interpretation of the
cooperative positioning problem in wireless sensor networks and
formulated the position problem as an implicit convex feasibility
problem.  We have then proposed a distributed algorithm based on
projections approach to solve the problem.  The proposed algorithm
enjoys parallel implementation capabilities and is suitable for
practical scenarios.  We have also proven that the algorithm converges
to the desired minimizer of an intuitively pleasing cost function,
\eqref{problem}, regardless if the implicit convex feasibility problem
is consistent or not.  Simulation results show an enhanced performance
of the proposed approach compared with available algorithms for sparse
networks.  {Since the proposed algorithm has low
  convergence speed, one possible open problem for future studies is
  to improve the convergence speed of the algorithm, while maintaining
  convergence.}

\section{Acknowledgments}
{We would like to thank the editor and anonymous
  reviewers for the valuable comments and suggestions which improved
  the quality of the paper.}  The work of M.~R.~Gholami and
E.~G.~Str\"om was  supported by the Swedish Research Council (contract
no.~2007-6363).  The work of L.~Tetruashvili and Y.~Censor was supported by United States-Israel Binational Science
Foundation (BSF) Grant number 200912 and by US Department of Army
Award number W81XWH-10-1-0170.

%
%
%
%
%
%

\footnotesize
\bibliographystyle{IEEEtran}

\bibliography{Ref}

\end{document}